\documentclass[submission]{dmtcs}

\usepackage[latin1]{inputenc}

\usepackage{amsmath}
\usepackage{graphicx}
\usepackage{caption}
\usepackage{amssymb}
\usepackage[ruled,vlined,algosection]{algorithm2e}
\usepackage{adjustbox}
\usepackage{multirow}
\usepackage{lscape}
\usepackage{rotating}
\usepackage{enumitem}

\newtheorem{theoremm}{Theorem}
\newtheorem{eqed}{Example}
\newtheorem {lemmaa}{Lemma}

\newtheorem {defnn}{Definition}
\newtheorem {corollaryy}{Corollary}
\newtheorem {conjecturee}[theoremm]{Conjecture}

\newtheorem {procd}{Procedure}

\newenvironment{example}{\begin{eqed}{\bf :}\sl}{\hfill\end{eqed}}
\newenvironment{lemma}{\begin{lemmaa} \sl}{\end{lemmaa}}
\newenvironment{theorem}{\begin{theoremm}{\bf :}\sl}{\end{theoremm}}
\newenvironment{definition}{\begin{defnn}{\bf :}\sl}{\end{defnn}}
\newenvironment{algo}{\begin{algorithm}}{\end{algorithm}}
\newenvironment{corollary}{\begin{corollaryy}{\bf :}\sl}{\end{corollaryy}}

\providecommand{\floor}[1]{\left \lfloor #1 \right \rfloor }

\author{Raju Hazari~~~~ and~~~~ Sukanta Das}
\title{On Number Conservation of Non-uniform Cellular Automata}
\address{Department of Information Technology,\\ Indian Institute of Engineering Science and Technology (IIEST), Shibpur, India \\ Email: \email{hazariraju0201@gmail.com  , sukanta@it.iiests.ac.in}}

\begin{document}
\maketitle
\begin{abstract}
This paper studies the number conservation property of 1-dimensional non-uniform cellular automata (CAs). In a non-uniform cellular automaton (CA), different cells may follow different rules. The present work considers that the cells follow Wolfram's CAs rules. A characterization tool, named Reachability tree is used to discover the number conservation property of non-uniform CAs. Then a decision algorithm is reported to conclude whether a given non-uniform CA with $n$ cells is number conserving or not. Finally, a synthesis scheme is developed to get an $n$-cell number conserving non-uniform CA. 
\end{abstract}

\section{Introduction}
\label{sec:in}

A cellular automaton (CA) is a regular lattice of simple finite state machines that update their states according to a local update rule. The local rule specifies the new state of each cell based on the states of its neighbors. Classically, cells of a CA follow same rule to generate their next states. Since late 1980s, however, a group of researchers had started to explore a new type of cellular automata (CAs) where different cells of a CA may follow different rules \cite{ppc1,sipper96evnca}. These CAs are commonly known as {\em non-uniform CAs}. Primary focus of this non-uniform CA research was on the one-dimensional CA, where the cells follow Wolfram's CA rules \cite{Wolfram86}. This work also considers non-uniform one-dimensional 3-neighborhood two-state (finite) CAs, and studies their number conservation property.

The number conserving cellular automata (NCCAs) are the CAs where the number of 1s (0s) of initial configuration is conserved during the evolution of the CA. Due to their similarity with the physical law of conservation, the NCCAs have received a great attention of the researchers in last two decades \cite{Hattori:1991, boccara-1998-31,Boccara:2002,Durand03}. The NCCAs essentially model the particle systems that are governed by local interaction rules. One such area where the NCCAs are widely utilized is the development of road traffic models \cite{NS92,FI96,SMC09,CSF11}.

In 1998, Boccara and Fuk\'{s} gave necessary and sufficient conditions for a one-dimensional CA to be NCCA, first for two states per cell \cite{boccara-1998-31} and then for an arbitrary number of states \cite{Boccara:2002}. Pivato \cite{Pivato:2002} gave a general treatment of conserved quantities in 1D CA, showing how to construct all the CAs which satisfy a given conservation law.
Durand et. al. showed that number-conservation is a decidable property through the generalization of the characterization of NCCA to 2 and then $d$ dimensions \cite{Durand03}. The idea of extending conserved quantities to that of monotone quantities was presented by Kurka \cite{Kurka:2003} where he considered CA with vanishing particles. Morita, Imai and other collaborators \cite{Morita98,Morita:99} use partitioned number-conserving (and reversible) CA. This is not exactly the same as an NCCA, since the reduction from a partitioned CA to a non-partitioned one does not preserve the number. In these works they show the computational universality of this class of automata, by simulating a counter machine which is known to be universal.

However, all the works on number conservation deal with classical CAs, where the cells of a CA follow a single rule. The concepts of deciding and constructing uniform (that is, classical) NCCAs can not be applied to non-uniform cases. In fact, there is no efficient algorithm to verify whether a given non-uniform CA is NCCA or to {\em synthesize} a non-uniform NCCA. This scenario motivates us to undertake this research. We develop here the theories for non-uniform NCCAs, and design efficient algorithms for deciding and synthesizing of non-uniform NCCA. By {\em synthesis}, however, here we mean the selection of individual cell rules of a non-uniform NCCA.

The paper is organized as follows. Section~\ref{RTncca} introduces a characterization tool named reachability tree for exploring non-uniform CAs, and identifies its role in characterizing non-uniform NCCAs. Section~\ref{property} further analyzes the behavior of non-uniform NCCAs, and based on this analysis, Section~\ref{algorithms} develops an efficient decision algorithm. Finally, a synthesis scheme is reported in Section~\ref{Synthesis_ncca} to construct an n-cell non-uniform NCCA.

Before progressing further, let us introduce some preliminary concepts and definitions about CAs.


\section{Cellular Automata Preliminaries}
\label{preli}

A cellular automaton (CA) is a discrete, spatially-extended dynamical system that has been studied extensively as a model of physical system. It evolves in discrete space and time. A CA consists of a lattice of cells, each of which stores a discrete variable at time $t$ that refers to the present state of the CA cell \cite{Neuma66}. The next state of a cell is affected by its present state and the present states of its neighbors at time $t$.
In two-state 3-neighborhood (self, left and right neighbors) 1-dimensional CA, next state of a cell is determined as:
\begin{equation}
S^{t+1}_i = f_i({S^t_{i-1}},{S^t_i},{S^t_{i+1}})
\end{equation}
where $f_i$ is the next state function of $i^{th}$ cell; ${S^t_{i-1}}$, ${S^t_i}$ and ${S^t_{i+1}}$ are the present states of the
left neighbor, self and right neighbor of the $i^{th}$ CA cell at time $t$.
Therefore, the function $f_i:\{0,1\}^3 \mapsto \{0, 1\}$ can be expressed as a look-up table.
The decimal equivalent of the 8 outputs is called `rule' \cite{wolfram86}.
Two such rules are 136 and 252 (Tab.~\ref{Trules}).

Traditionally, each of the cells of a CA follows same next state function. Such a CA is called as {\em uniform} CA. On the other hand, if the CA cells are allowed to follow different next state functions (rules), the CA is a {\em non-uniform} (or  {\em hybrid}) CA. For an $n$-cell hybrid CA, we need a {\em rule vector} ${\mathcal{R}}= \langle {\mathcal{R}}_0, {\mathcal{R}}_1, \cdots, {\mathcal{R}}_i, \cdots, {\mathcal{R}}_{n-1}\rangle$, where the CA cell $i$ ($0\le i\le n-1$) uses rule ${\mathcal{R}}_i$. In case of an $n$-cell uniform CA,  ${\mathcal{R}}_0= {\mathcal{R}}_1= \cdots = {\mathcal{R}}_i = \cdots = {\mathcal{R}}_{n-1}$. Hence, traditional uniform CA is a special case of non-uniform CA. This work deals with finite binary CAs of size $n$ under {\em periodic boundary} condition where first and last cells are neighbors of each other.

{\small
\begin{table*}
\caption{Look-up table for rule 136, 252, 238 and 226}
\label{Trules}
\[
\begin{array}{|cccccccccc|}
\hline
{\rm Present~ State:}    &  111 & 110 & 101 & 100 & 011 &  010 &  001 &  000 & Rule \\
(RMT)	& (7) & (6) & (5) & (4) & (3) & (2) & (1) & (0) & \\
  {\rm ~ (i)~ Next~ State:}    &   1  &  0 &  0  &  0  &   1  & 0  &   0  &   0  & 136 \\
{\rm (ii)~ Next~ State:}    & 1  &  1 &  1  &  1  &   1  & 1  &   0  &   0   & 252 \\
{\rm (iii)~ Next~ State:}    & 1  &  1 &  1  &  0  &   1  & 1  &   1  &   0   & 238 \\
{\rm (iv)~ Next~ State:}    & 1  &  1 &  0  &  0  &   0  & 0  &   0  &   0   & 192 \\
\hline
\end{array}
\]
\end{table*}
}

A collection of states of the cells ${\mathcal S}^t = (S^t_0,S^t_1,\cdots,S^t_{n-1})$ at time $t$
is the present configuration or (global) state of the CA. In case of a NCCA, the number of 1s in any seed remains unaltered during evolution of the CA. This implies, for each pair of ${\mathcal S}^t$ and ${\mathcal S}^{t+1}$, the number of 0s and 1s in ${\mathcal S}^t$ remain unchanged in ${\mathcal S}^{t+1}$. 

\begin{figure}[h]
\begin{center}
\includegraphics[height=1.2in]{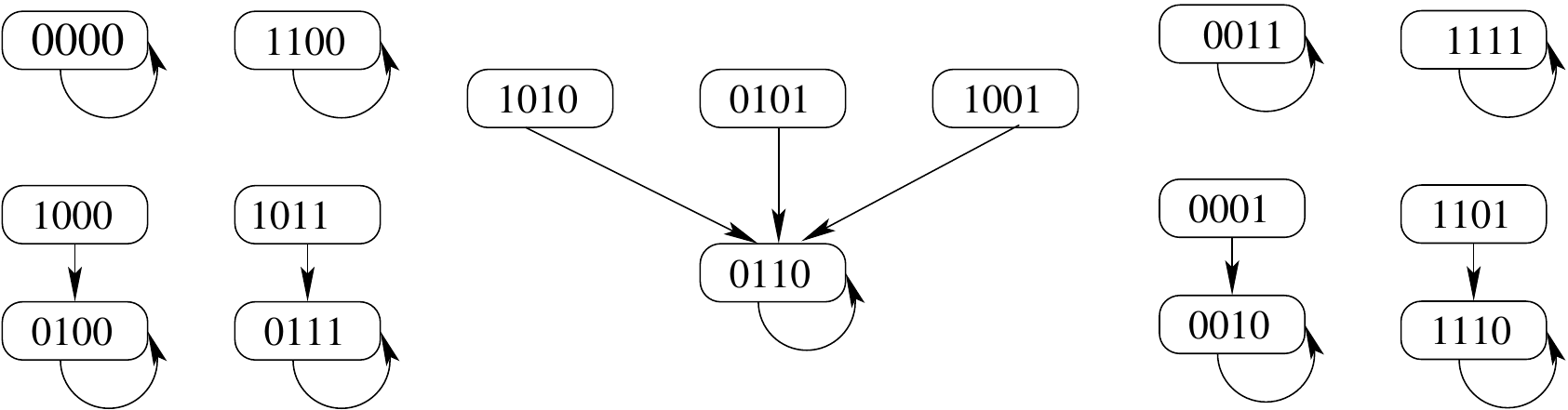}
\caption{State transition diagram of the CA $\langle136, 252, 238, 192\rangle$}
\label{NCCA4}
\end{center}
\end{figure}

Fig.~\ref{NCCA4} shows the state transition diagram of a 4-cell CA with rule vector $\langle136, 252, 238, 192\rangle$. The number of 0s and 1s of each state remain unchanged in its next state. For example, number of 0s and 1s of state 1010 are conserved in its next state 0110 (Fig.~\ref{NCCA4}). Hence, it is an NCCA. The state transition diagram of a CA classifies its states as {\em reachable} and {\em non-reachable}. A state is reachable if it has at least one predecessor. On the other hand, a non-reachable state can not be reached from any state. The states 1000, 1011, 1010, 0101, 1001, 0001 and 1101 of Fig.~\ref{NCCA4} are non-reachable. The rest states of Fig.~\ref{NCCA4} are reachable.

As mentioned before, the rules can be expressed in tabular form (Tab.~\ref{Trules}). Note that the table has an entry for each value of $S^t_{i-1}$, $S^t_i$  and $S^t_{i+1}$. In this work, we call the combination of present state as {\em Rule Min Term} (RMT) because this representation can be viewed as {\em Min Term} of three variable {\em Switching function}. For our convenience, we generally represent RMTs by their corresponding decimal equivalents. The RMTs have binary values (0/1) which correspond to the next states for these RMTs. For example, the RMT 011 (RMT 3) in Tab.~\ref{Trules} has the next state value 1 for rule 136 and 0 for rule 192. We write the next state of an RMT $r$ of a rule ${\mathcal R}_i$ as ${\mathcal R}_i[r]$. Hence, 136[3] = 1 and 192[3] = 0.

%

A CA state can also be viewed as a sequence of RMTs (RS). For example, the state 1110 in periodic boundary condition can be viewed as $\langle3765\rangle$, where 3, 7, 6 and 5 are corresponding RMTs on which the transition of first, second, third and forth cells can be made. For an $n$-bit state, we get a sequence of $n$ RMTs. To get an RMT sequence of a state, we consider an imaginary 3-bit window that slides over that state. To get the $i^{th}$ RMT of the sequence, the window is loaded with the $i-1$, $i$ and $(i+1)^{th}$ bits of the state. The window slides one bit right to report the $(i+1)^{th}$ RMT of the sequence. During the finding of RS, however, one can observe a relation between two consecutive RMTs of a sequence. If $i^{th}$ RMT is 0, then $(i+1)^{th}$ RMT can either be 0 or 1. Similarly, if the $i^{th}$ RMT is 3 or 7, $(i+1)^{th}$ RMT is either 6 or 7. In general, if $r$ is the $i^{th}$ RMT, then $(i+1)^{th}$ RMT is $2r \pmod{8}$ or $(2r+1) \pmod{8}$. Such a relation of RMTs is shown in Tab.~\ref{possible_weight}.

\begin{defnn}
\label{equi}
Two RMTs $r$ and $s$ ($r$ $\neq$ $s$) are said to be equivalent to each other if $2r \equiv 2s     \pmod{8}$.
\end{defnn}

The RMTs 0 and 4 of ${\mathcal R}_i$ are equivalent to each other(Tab.~\ref{relation}). Similarly, RMTs 1 \& 5, 2 \& 6, and 3 \& 7 are the equivalent RMT pairs.

\begin{defnn}
\label{sibl}
Two RMTs $r$ and $s$ ($r$ $\neq$ $s$) are said to be sibling RMT if $\lfloor \frac{r}{2} \rfloor$ = $\lfloor \frac{s}{2} \rfloor$
\end{defnn}

The RMTs 0 and 1 of ${\mathcal R}_i$ are sibling of each other. Similarly, RMTs 2 \& 3, 4 \& 5 and 6 \& 7 are the sibling RMT pairs.

We represent $Equi_i$ as a set of RMTs that contains RMT $i$ and all of its {\em equivalent} RMTs. That is, $Equi_i$ = \{$i$, 4+$i$\}, where 0 $\leq$ $i$ $\leq$ 3. Similarly, $Sibl_j$ represent a set of sibling RMTs where $Sibl_j$ = \{2j, 2j+1\} and 0 $\leq$ $j$ $\leq$ 3. The $Equi_i$ and $Sibl_j$ maintain an interesting relationship. Tab.~\ref{relation} elaborates this relationship. We use these relations in next sections.

{\small
\begin{table}[h]
\caption{Relationship among the RMTs for 2-state CA}
\begin{center}
\label{relation}
\begin{tabular}{|c|c|c||c|c|c|}\hline
 \multicolumn{3}{|c||}{RMT at $i^{th}$ rule}&
   \multicolumn{3}{|c|}{RMTs at $(i+1)^{th}$ rule}\\\hline
\# Set & Equivalent & Decimal & \# Set & Sibling & Decimal\\
 & RMTs & Equivalent & & RMTs & Equivalent \\\hline
$Equi_0$ & 000, 100 & 0, 4 & $Sibl_0$ & 000, 001 & 0, 1 \\\hline
$Equi_1$ & 001, 101 & 1, 5 & $Sibl_1$ & 010, 011 & 2, 3\\\hline
$Equi_2$ & 010, 110 & 2, 6 & $Sibl_2$ & 100, 101 & 4, 5\\\hline
$Equi_3$ & 011, 111 & 3, 7 & $Sibl_3$ & 110, 111 & 6, 7\\\hline
\end{tabular}
\end{center}
\end{table}
}

\section{Reachability tree and Number conservation}
\label{RTncca}

Reachability tree \cite{Acri04,SukantaTh}, is a discrete tool of characterizing 1-dimensional CA. It is a rooted and edge-labelled binary tree that represents reachable states of a binary CA. For an $n$-cell CA, there are $(n+1)$ levels - root at level 0, and leaves at level $n$. 

\begin{defnn}
\label{Rtree_def}
Reachability tree for an $n$-cell cellular automaton under periodic boundary condition is a rooted and edge-labeled binary tree with $n+1$ levels, where
each node $ N_{i.j} ~ (0 \leq i \leq n,~ 0 \leq j \leq 2^{i}-1)$ is an ordered list of $4$ sets of RMTs, and the root $N_{0.0}$ is the ordered list of all sets of sibling RMTs. We denote the edges between $N_{i.j} ~ (0 \leq i \leq n-1,~ 0 \leq j \leq 2^{i}-1)$ and its children as  $E_{i.2j}=(N_{i.j},N_{i+1.2j},l_{i.2j})$ and $E_{i.2j+1}=(N_{i.j},N_{i+1.2j+1},l_{i.2j+1})$ where $l_{i.2j}$ and $l_{i.2j+1}$ are the labels of the edges. Like nodes, the labels are also ordered list of $4$ sets of RMTs. Let us consider that ${\Gamma_{p}}^{N_{i.j}}$ is the $p^{th}$ set of the node $N_{i.j}$, and ${\Gamma_{q}}^{E_{i.2j+m}}$ is the $q^{th}$ set of the label on edge $E_{i.2j+m}$ $(0 \leq p,q \leq 3)$. So,  $N_{i.j} = ( {\Gamma_{p}}^{N_{i.j}})_{0 \leq p \leq 3}$ and $l_{i.2j+m} =  ( {\Gamma_{q}}^{E_{i.2j+m}})_{0 \leq q \leq 3}$ ($m \in \{0, 1\}$). Following are the relations which exist in the tree :

\begin{enumerate}
	\item \label{rtd1} [For root] $N_{0.0} = ({\Gamma_{k}}^{N_{0.0}})_{0 \leq k \leq 3}$, where $\Gamma_0^{N_{0.0}}=\{0,1\}$, $\Gamma_1^{N_{0.0}}=\{2,3\}$, $\Gamma_2^{N_{0.0}}=\{4,5\}$ and  $\Gamma_3^{N_{0.0}}=\{6,7\}$.
	
	\item \label{rtd2} $\forall r \in {\Gamma_{k}}^{N_{i.j}}$, RMT $r$ of ${\mathcal R}_{i}$  is included in ${\Gamma_{k}}^{E_{i.2j}}$ (resp. ${\Gamma_{k}}^{E_{i.2j+1}}$), if ${\mathcal R}_{i}[r]$ = 0 (resp. 1), where ${\mathcal R}_{i}$ is the rule of the $i^{th}$ cell of the CA. That means, $ {\Gamma_{k}}^{N_{i.j}} = {\Gamma_{k}}^{E_{i.2j}}$ ${\bigcup}$ ${\Gamma_{k}}^{E_{i.2j+1}}$ $(0 \leq k \leq 3)$.

	\item \label{rtd4}$\forall r \in {\Gamma_{k}}^{E_{i.2j}}$ (resp. ${\Gamma_{k}}^{E_{i.2j+1}}$), RMTs $2r \pmod{8}$ and $ 2r+1 \pmod{8}$ of ${\mathcal R}_{i+1}$ are in ${\Gamma_{k}}^{N_{i+1.2j}}$ (resp. ${\Gamma_{k}}^{N_{i+1.2j+1}}$).

	\item \label{rtd5} [For level $n-2$] ${\Gamma_{k}^{N_{n-2.j}}} = \lbrace s ~|~$ if $ r \in {\Gamma_{k}^{E_{n-3.j}}} $ then $ s \in \lbrace 2r \pmod{8}, 2r+1 \pmod{8}\rbrace \cap \lbrace i, i+2, i+4, i+6 \rbrace \rbrace$ ($i= \floor{\frac{k}{2}}, 0 \leq k \leq 3, 0 \leq j \leq 2^{n-2}-1 $).

	\item \label{rtd6} [For level $n-1$] ${\Gamma_{k}}^{N_{n-1.j}} = \lbrace s ~|~$ if $ r \in {\Gamma^{E_{n-2.j}}}_{k} $ then $ s \in \lbrace 2r \pmod{8}, 2r+1 \pmod{8}\rbrace \cap \lbrace k, k+4 \rbrace \rbrace, { 0 \leq k \leq 3}$.
	
\end{enumerate}

\end{defnn}

Note that the nodes of levels $n-2$ and $n-1$ are different from other intermediate nodes (Points~\ref{rtd5} and \ref{rtd6} of Definition~\ref{Rtree_def}). Only $\frac{1}{2}$ of selective RMTs can play as $ {\Gamma_{k}}^{N_{i.j}}$ in a node $ {N_{i.j}}, (0\leq k \leq 3, 0 \leq j \leq 2^i-1)$ when $i = n-2$ or $n-1$. Finally, we get the leaves with ${\Gamma_{k}}^{N_{n.j}}$, where ${\Gamma_{k}}^{N_{n.j}}$ is either empty or a set of sibling RMTs. Note that, ${\Gamma_{k}}^{N_{0.0}}$ is a set of sibling RMTs (Point~\ref{rtd1} of Definition~\ref{Rtree_def}) and ${\Gamma_{k}}^{N_{0.0}} = \bigcup _{j} \Gamma_k^{N_{n.j}} $. That is, the RMTs of leaves are the RMTs of ${\mathcal R}_0$. We call $E_{i.2j}$ as 0-edge, because the next state values of all the RMT at the label of $E_{i.2j}$ are 0 (Point~\ref{rtd2} of Definition~\ref{Rtree_def}). Similarly, $E_{i.2j+1}$ is called as 1-edge. If not stated otherwise, ``${N_{i.j}}$'' will mean an arbitrary node in the reachability tree of an $n$-cell CA in our further discussion, where $ 0 \leq i \leq n$ and $0 \leq j \leq 2^n-1$.

\begin{figure}
\begin{center}
\includegraphics[height=2.5in, width=5.6in]{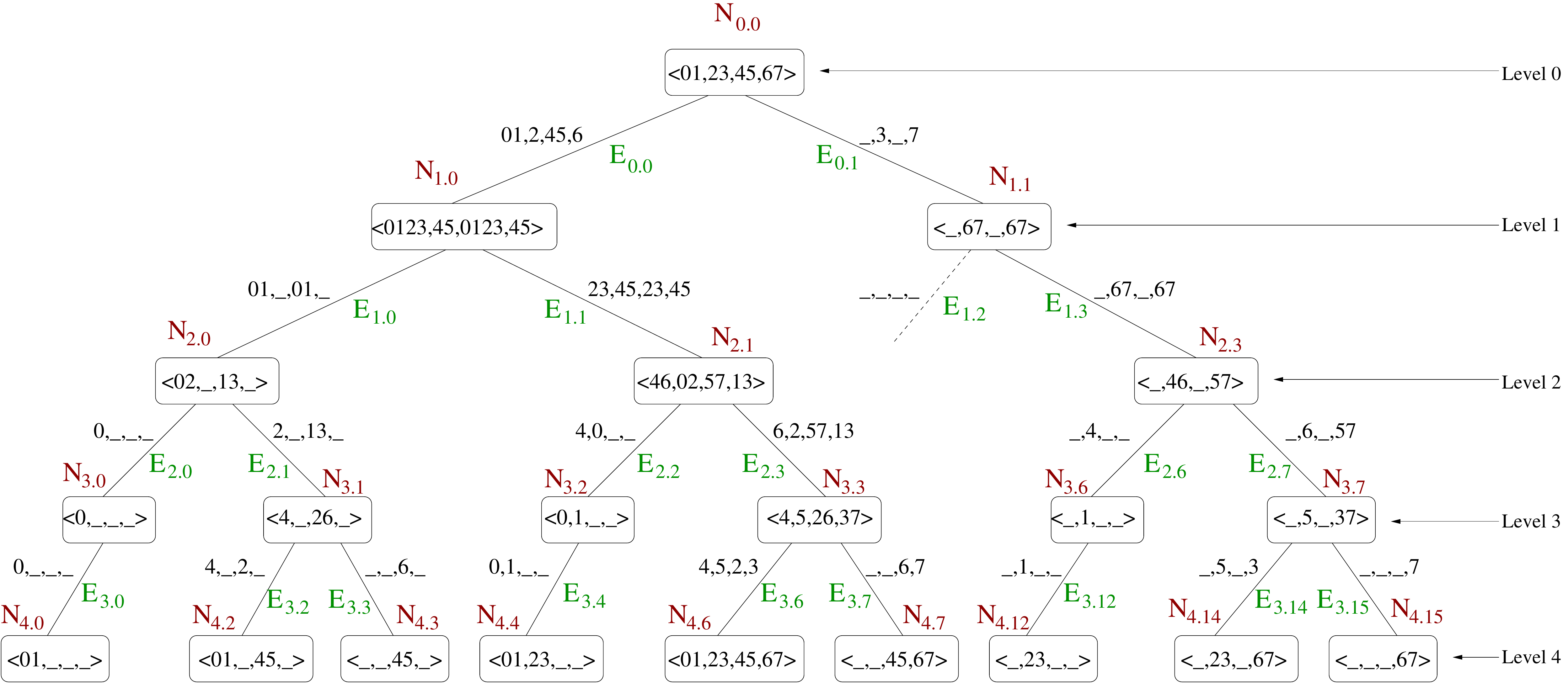}
\caption{Reachability Tree for the CA $\langle136, 252, 238, 192\rangle$}
\label{treestructureP}
\end{center}
\end{figure}

Consider a CA with rule vector $\langle136, 252, 238, 192\rangle$ (Fig.~\ref{NCCA4}). The RMTs of CA rules are noted in Tab.~\ref{Trules}. The reachability tree of the CA is shown in Fig.~\ref{treestructureP}. The root of the tree is $N_{0.0}$ = ($\Gamma_0^{N_{0.0}}$, $\Gamma_1^{N_{0.0}}$, $\Gamma_2^{N_{0.0}}$ , $\Gamma_3^{N_{0.0}}$). Note that the root is independent of CA rule, whereas other nodes are rule dependent. Here, $l_{0.0}$, the label of $E_{0.0}$, is ($\{0,1\}, \{2\}, \{4,5\}, \{6\}$) and the corresponding child $N_{1.0}$ is ($\{0,1,2,3\}, \{4,5\}, \{0,1,2,3\}, \{4,5\}$) (Point~\ref{rtd4} of Definition~\ref{Rtree_def}).
If no RMT is the member of $\Gamma_k^{N_{i.j}}$ ($0 \leq k \leq 3$, $0 \leq i \leq n-1$, $0\leq j \leq  2^i-1$) (i.e, the set is empty), then $\Gamma_k^{N_{i.j}}$ is noted as `$\_$' in the tree. For example, $l_{0.1}$, the label of $E_{0.1}$ is ($\emptyset$, \{3\}, $\emptyset$, \{7\}) and the corresponding child $N_{1.1}$ is ($\emptyset$, \{6,7\}, $\emptyset$, \{6,7\}). However, an arbitrary RMT can not be a part of nodes of level $n-2$ and $n-1$. For example, $l_{1.0}$ = (\{0,1\}, $\emptyset$, \{0,1\}, $\emptyset$) but the corresponding node $N_{2.0}$ (that is, $N_{n-2.1}$, since $n$ = 4) is (\{0,2\}, $\emptyset$, \{1,3\}, $\emptyset$ ). Observe that $\Gamma_2^{N_{2.0}}$ of the node $N_{2.0}$ of Fig.~\ref{treestructureP} is \{1,3\}. If we follow Point~\ref{rtd4} of Definition~\ref{Rtree_def}, then RMTs 0 and 2 should also be part of $\Gamma_2^{N_{2.0}}$. But they could not be, because the node is at level $n-2$ (Point~\ref{rtd5} of Definition~\ref{Rtree_def}). Similarly, $l_{2.3}$ = (\{6\}, \{2\}, \{5,7\}, \{1,3\}) and $N_{3.3}$ = (\{4\}, \{5\}, \{2,6\}, \{3,7\}) (Point~\ref{rtd6} of Definition~\ref{Rtree_def}). Here, $\Gamma_k^{N_{0.0}}$ = $\cup_j \Gamma_k^{N_{4.j}}$ for each $k \in \{0, 1, 2, 3\}$.

Reachability tree gives us information about reachable states of the CA. However, some nodes in a reachability tree may not be present, which we call non-reachable nodes, and the corresponding missing edges are non-reachable edges. For them $\Gamma_k^{N_{i.j}}$ = $\emptyset$ and $\Gamma_k^{E_{i.j}}$ = $\emptyset$ for each $k \in \{0, 1, 2, 3\}$. For example, in Fig.~\ref{treestructureP}, the edge $E_{1.2}$ and corresponding node $N_{2.2}$ are not present. The dotted edge indicates the non-reachable edge.

A sequence of edges $\langle  E_{0.j_0},E_{1.j_1},\cdots,E_{i.j_i},E_{i+1.j_{i+1}}, \cdots, E_{n-1.j_{n-1}} \rangle$ from root to a leaf associates a reachable state and at least one RMT Sequence (RS) $\langle r_0r_1 \cdots r_ir_{i+1} \cdots r_{n-1} \rangle$, where $r_i \in \Gamma_p^{E_{i.j_i}}$ for any $p \in \{0, 1, 2, 3\}$ and $r_{i+1} \in \Gamma_q^{E_{i+1.j_{i+1}}}$ for any $q \in \{0, 1, 2, 3\}$ $(0 \leq i < $n$-1, 0 \leq j_i \leq 2^i-1,$ and $j_{i+1} = 2j_i$ or $2j_i+1)$. That is, a sequence of edges represents at least two CA states. Note that if ${\mathcal R}_i[r_i]$ = 0 (resp. 1), then $E_{i.j_i}$ is 0-edge (resp. 1-edge). Therefore, the reachable state is the next (resp. present) state of the current (resp. predecessor) state, represented as RMT sequence. Interestingly, there are $2^n$ RSs in the tree, but number of reachable states may be less than $2^n$. However, a sequence of edges may associate $m$-number of RSs ($m\geq 1$), which implies, this state is reachable from $m$-number of different states. For example, the edge sequence $\langle E_{0.0},E_{1.1},E_{2.2},E_{3.4}\rangle$ of Fig.~\ref{treestructureP} represents the reachable state 0100, and associates two RMT sequences 1240 and 2401, which are the states 0100 and 1000 respectively. Both the states are the predecessors of 0100 (see Fig.~\ref{NCCA4} for verification). All of the $2^n$ states of an $n$-cell CA are present in the tree in the form of RMT sequences, and their next states as reachable states.

Reachability tree can be utilized for deciding an $n$-cell CA as NCCA. Following is our approach: Develop reachability tree of the CA to get the RMT sequences. If we observe that any state, represented as RMT sequence, is not having same number of 1s as it has in its next state, we decide the CA as not an NCCA. To facilitate this task, we assign {\em weights} to RMTs. Before defining weights formally, we report the following result.

\begin{lemmaa}
\label{universal}
For each rule ${\mathcal R}_i$ of an $n$-cell NCCA, ${\mathcal R}_i[0]$ = 0 and ${\mathcal R}_i[7]$ = 1.
\end{lemmaa}

\begin{proof}
Since the number of 1s (and 0s) of any initial state is conserved in NCCA, two homogeneous states $0^n$ and $1^n$ can not follow this condition without having ${\mathcal R}_i[0]$ = 0 and ${\mathcal R}_i[7]$ = 1 for each rule ${\mathcal R}_i$. 
\end{proof}

\noindent
\textbf{The Weight:}
The weight of an RMT notes the surplus or deficiency of 1s in the RMT sequence with respect to its next state.
Let us consider an RMT sequence $\langle r_0r_1 \cdots r_ir_{i+1} \cdots r_{n-1} \rangle$ and its next state $b_0 b_1 \cdots b_i \cdots b_{n-1}$ for a CA ${\mathcal R}$ = $\langle {\mathcal R}_0, {\mathcal R}_1, \cdots, {\mathcal R}_i, \cdots, {\mathcal R}_{n-1}\rangle$. That is, ${\mathcal R}_i[r_i] = b_i$, $i \in \{0, 1, \cdots, n-1\}$. Initially, there is no surplus or deficiency of 1s. So, weight of $r_0$ is 0. However, when we compare $r_0$ and $b_0$, then we can understand whether we have surplus or deficiency of 1. For example, if $r_0 \in \{2, 3, 6, 7\}$ (that is, the first bit of the state, represented as $\langle r_0 r_1 \cdots r_i \cdots r_{n-1} \rangle$ is 1) and $b_0$ = 0, then we understand that the first bit of the state $r_0 r_1 \cdots r_i \cdots r_{n-1}$ contributes one additional 1, which is carried by the RMR $r_1$. We assign weight 1 to RMT $r_1$. That is, if we have surplus (resp. deficiency) of 1, we assign weight to $r_1$ as 1 (resp. -1). Next we compare $r_1$ and $b_1$, and assign weight to $r_2$ after taking the weight of the $r_1$ into account. For example, if $r_1 \in \{2, 3, 6, 7\}$ and $b_1$ = 0, and weight of $r_1$ is 1, then weight of $r_2$ is 2. In this case, first two bits of state $r_0 r_1 \cdots r_i \cdots r_{n-1}$ are 11, but those of next state are 00. In this way we proceed, and finally we get the weight of RMT $r_{n-1}$. In case of NCCA, this weight is to be dismissed after comparing $r_{n-1}$ and $b_{n-1}$, to conclude that both the states have equal number of 1s. As an example, consider two consecutive states, RMT sequence $\langle4012\rangle$ and its next state 0010 of the CA of Fig.~\ref{NCCA4}. Weights of RMTs of the sequence are 0, 0, 0, -1 respectively. The weight -1 is dismissed when we compare RMT 2 and 0, and conclude that both the states have same number of 1s.

The reachability tree can efficiently implement this idea. To do so, we assign weights of RMTs at the nodes of reachability tree. Initially, weights of RMTs at root are 0. An RMT $r_i$, present in $\Gamma_k^{N_{i.j}}$ is also present at a label of an edge- either in $\Gamma_k^{E_{i.2j}}$ or in $\Gamma_k^{E_{i.2j+1}}$ (Point~\ref{rtd2} of Definition~\ref{Rtree_def}). If the $r_i$ is in $\Gamma_k^{E_{i.2j}}$ (resp. $\Gamma_k^{E_{i.2j+1}}$), the next state of $r_i$ for rule ${\mathcal{R}}_{i}$ is 0 (resp. 1). So the weight of RMT $r_{i+1}$ can be understood by knowing the weight of $r_i$ and observing the edge ($E_{i.2j}$ or $E_{i.2j+1}$) on which the RMT exists. However, RMTs $r_i$ and $r_{i+1}$ are in an RMT sequence, and they are on some nodes as well. Suppose $W_k(i, r)$ denotes the weight of RMT $r \in \Gamma_k^{N_{i.j}}$. Following is the definition of the weight, for each $k \in \{0, 1, 2, 3\}$.

\begin{enumerate}
\item  \hspace{2.2in}$W_k(0, r)$ = 0 ~~~~~~~~~:$\forall r \in \Gamma_k^{N_{0.0}}$
\item	
\begin{small}
 \begin{displaymath}
 W_k(i, 2r\pmod{8}) = W_k(i, 2r+1\pmod{8}) = \left\{
 \begin{array}{lr}
 W_k(i-1, r) + 1  & : r \in \{2, 3, 6, 7 \} ~~and~~  {\mathcal{R}}_{i-1}[r] =0 \\
 W_k(i-1, r) - 1  & : r \in \{0, 1, 4, 5 \} ~~and~~  {\mathcal{R}}_{i-1}[r] =1 \\
 W_k(i-1, r) & :   otherwise 
 \end{array}
 \right.
\end{displaymath} 
\end{small} 
\end{enumerate}

We use the above definition to get the weights of the RMTs, present in nodes. Since the RMTs at leaves are the RMTs of ${\mathcal R}_0$, that is, of root, the weight of such an RMT $r$ is the additional 1s which the state, represented as RMT sequence that involves this $r$, is having compared to its next state. Therefore, if all the weights of leaves are 0, the CA is an NCCA.

\begin{example}
Consider the CA with rule vector $\langle136, 252, 238, 192\rangle$ (Tab.~\ref{Trules}). We assign the weights to the RMTs at nodes according to the above rule. The reachability tree with such weight is noted in Fig.~\ref{NCCAweight}. The weights are shown in the bottom of the RMTs (within first brackets). The weights of the RMTs of the sequence $\langle4012\rangle$ are 0, 0, 0, and -1 respectively. The weight -1 is nullified at the leaf $N_{4.2}$, as 192[2] = 0. Many of the RMTs of intermediate nodes have non-zero weights. However, the leaves of the tree have RMTs with zero weight. Hence, the CA is an NCCA.
\end{example}

\begin{figure}
\includegraphics[height=2.3in,width=5.5in]{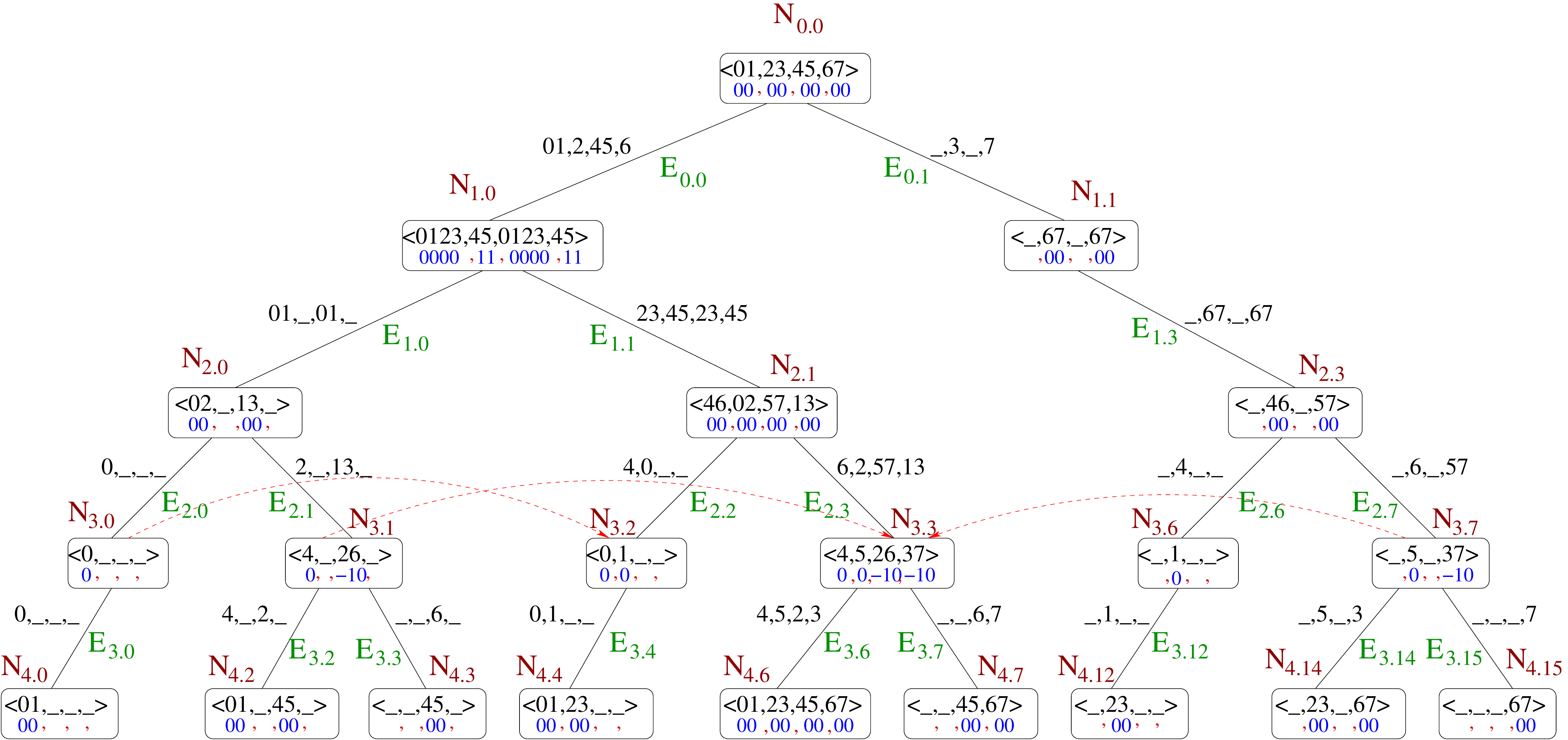}
\caption{Reachability tree of NCCA with weights of RMTs at different nodes}
\label{NCCAweight}
\end{figure}

According to Point~\ref{rtd4} of Definition~\ref{Rtree_def}, an RMT $r \in \Gamma_k^{E_{i.j}}$ contributes two RMTs~ $2r\pmod{8}$ and $2r+1\pmod{8}$ to $\Gamma_k^{N_{i+1.j}}$. If another RMT $s$, which is equivalent to $r$, present in $\Gamma_k^{E_{i.j}}$, two more RMTs $2s\pmod{8}$ and $2s+1\pmod{8}$ are contributed to $\Gamma_k^{N_{i+1.j}}$ where $2r$ $\equiv$ $2s\pmod{8}$. However, the weights of the RMTs $2r\pmod{8}$ and $2s\pmod{8}$ may not be same. If this happens, the CA can not be an NCCA. Following is an important result.

\begin{theorem}
\label{basic1}
The weight of an RMT in $\Gamma_k^{N_{i.j}}$ is unique, where $k \in \{0, 1, 2, 3\}$ and $N_{i.j}$ is any node in the reachability tree of an NCCA.
\end{theorem}

\begin{proof}
Suppose an RMT $r \in \Gamma_k^{N_{i.j}} $ has more than one weight. Then, the RMTs in $\Gamma_k^{N_{i+1.2j}}$ or in $\Gamma_k^{N_{i+1.2j+1}}$ contributed by $r$ (using Point~\ref{rtd4} of Definition~\ref{Rtree_def}) have more than one weight. Again these RMTs contribute RMTs with multiple weights to the next level, and so on. Finally in the leaf nodes of the reachability tree, we get an RMT having multiple weights. But to be NCCA, in the leaf nodes all RMTs should have weight 0. So, this violates the condition of NCCA. Hence, the weight of an RMT in $\Gamma_k^{N_{i.j}}$ is to be unique.
\end{proof}

Following corollaries can be derived from Theorem~\ref{basic1} for an NCCA.

\begin{corollary}
\label{basic2}
The weights of RMT $r$ of ${\mathcal R}_i$ that belongs to $\Gamma_k^{N_{i.a}}$ and $\Gamma_k^{N_{i.b}}$ are equal.
\end{corollary}

\begin{proof}
For a proof by contradiction, assume that $r$ has two different weights in two nodes. Hence, the RMTs contributed by $r$ (following Point~\ref{rtd4} of Definition~\ref{Rtree_def}) to the nodes of next level have different weights. As a result, at least one set of leaves followed from $N_{i.a}$ or $N_{i.b}$ can have RMTs with weight not equal to 0. The RMTs with non-zero weight in any leaves indicate that the CA is not an NCCA. Hence to be NCCA, the RMT $r$ has to have same weight in two nodes.
\end{proof}

\begin{corollary}
\label{equivalentTh1}
Two equivalent RMTs $r \in \Gamma_k^{N_{i.a}}$ and $s \in \Gamma_k^{N_{i.b}}$ ($k \in \{0, 1, 2, 3\}$) carry same weight when ${\mathcal R}_i[r]$ = ${\mathcal R}_i[s]$.
\end{corollary}

\begin{proof}
Since ${\mathcal R}_i[r]$ = ${\mathcal R}_i[s]$, both the RMTs either on 0-edge or on 1-edge. Further, since the RMTs are equivalent to each other (i.e, $2r \equiv 2s \pmod{8}$), they contribute same set of (sibling) RMTs to $\Gamma_k^{N_{i+1.j_1}}$ ($j_1$ = $2a$ or $2a+1$) and $\Gamma_k^{N_{i+1.j_2}}$ ($j_2$ = $2b$ or $2b+1$)(Point~\ref{rtd4} of Definition~\ref{Rtree_def}). Hence, the weights of RMTs $r$ and $s$ are to be same, otherwise Theorem~\ref{basic1} and corollary~\ref{basic2} will be violated.
\end{proof}

\begin{corollary}
\label{1110}
For two equivalent RMTs $r \in \Gamma_k^{N_{i.a}}$ and $s \in \Gamma_k^{N_{i.b}}$ ($k \in \{0, 1, 2, 3\}$),  ${\mathcal R}_i[r] \neq {\mathcal R}_i[s]$ when $W_k(i, r) \neq W_k(i, s)$. 
\end{corollary}

\begin{proof}
Here, RMTs $r$ and $s$ contribute same set of RMTs to the nodes followed from ${N_{i.a}}$ and ${N_{i.b}}$. If ${\mathcal R}_i[r] = {\mathcal R}_i[s]$ when $W_k(i, r) \neq W_k(i, s)$, the weights of RMT $2r \in \Gamma_k^{N_{i.a}}$ and $2s \in \Gamma_k^{N_{i.b}}$ where $2r=2s$ are not unique, which violates Theorem~\ref{basic1}. Hence ${\mathcal R}_i[r] \neq {\mathcal R}_i[s]$ when $W_k(i, r) \neq W_k(i, s)$.
\end{proof}

\begin{corollary}
\label{1111}
$|W_k(i, r) - W_k(i, s)|$ = 0 or 1 where $r \in \Gamma_k^{N_{i.a}}$ and $s \in \Gamma_k^{N_{i.b}}$ ($k \in \{0, 1, 2, 3\}$) are two equivalent RMTs.
\end{corollary}

\begin{proof}
According to Corollary~\ref{equivalentTh1}, $W_k(i, r)$ = $W_k(i, s)$ when ${\mathcal R}_i[r]$ = ${\mathcal R}_i[s]$. But when ${\mathcal R}_i[r]$ $\neq$ ${\mathcal R}_i[s]$ then in case of NCCA, difference between weights of RMTs $r$ and $s$ is 1, because $2r \pmod {8}$ = $2s \pmod {8}$ and so $W_k(i+1, 2r \pmod {8})$ = $W_k(i+1, 2s \pmod {8})$. Hence, $|W_k(i, r) - W_k(i, s)|$ = 0 or 1.
\end{proof}

\begin{corollary}
\label{basic3}
For two equivalent RMTs $r$ and $s$ that belong to $\Gamma_{k_1}^{N_{i.a}}$ and $\Gamma_{k_2}^{N_{i.b}}$ respectively where $k_1, k_2 \in \{0, 1, 2, 3\}$ and $k_1 \neq k_2$, $W_{k_1}(i, r) - W_{k_1}(i, s)$ = $W_{k_2}(i, r) - W_{k_2}(i, s)$.
\end{corollary}

\begin{proof}
Consider for $k_1 \neq k_2$, $W_{k_1}(i, r) - W_{k_1}(i, s)$ $\neq$ $W_{k_2}(i, r) - W_{k_2}(i, s)$. According to Corollary~\ref{1111}, $|W_{k_1}(i, r) - W_{k_1}(i, s)|$ = 0 or 1. So, if $|W_{k_1}(i, r) - W_{k_1}(i, s)|$ = 0 then $|W_{k_2}(i, r) - W_{k_2}(i, s)|$ = 1. Now according to Corollary~\ref{equivalentTh1}, ${\mathcal R}_i[r]$ = ${\mathcal R}_i[s]$ when $|W_{k_1}(i, r) - W_{k_1}(i, s)|$ = 0 and according to Corollary~\ref{1110}, ${\mathcal R}_i[r]$ $\neq$ ${\mathcal R}_i[s]$ when $|W_{k_2}(i, r) - W_{k_2}(i, s)|$ = 1. Hence, a contradiction. So, $|W_{k_1}(i, r) - W_{k_1}(i, s)|$ and $|W_{k_2}(i, r) - W_{k_2}(i, s)|$ are same. Hence proved.
\end{proof}

Next we define {\em sub-nodes}, the irrelevant nodes in deciding NCCA.

\begin{definition}
A node $N_{i.p}$ is \textit{sub-node} of another node $N_{i.q}$ if $\Gamma_k^{N_{i.p}} \subseteq \Gamma_k^{N_{i.q}}$ for each $k \in \{0, 1, 2, 3\}$. We write $N_{i.p} \subseteq N_{i.q}$.
\end{definition}

In Fig.~\ref{NPCAweight}, $N_{3.0}$ is the sub-node of $N_{3.2}$ ($N_{3.0}\subseteq N_{3.2}$). Similarly, $N_{3.1}\subseteq N_{3.3}$ and $N_{3.7}\subseteq N_{3.3}$. 

It is obvious from Corollary~\ref{basic1} that an RMT $r \in \Gamma_k^{N_{i.p}}$ has the same weight as it is in $\Gamma_k^{N_{i.q}}$ for any $k \in \{0, 1, 2, 3\}$, when $N_{i.p} \subseteq N_{i.q}$. This implies, if the weights of RMTs at leaves of the sub-tree rooted at ${N_{i.q}}$ are 0, then the weights of RMTs at leaves of the sub-tree rooted at $N_{i.p}$ are also 0. Hence, sub-nodes have no additional role in deciding NCCA, and the nodes excluding their sub-nodes can successfully decide whether the CA is NCCA. During the construction of the tree, therefore, if a sub-tree of a node is found, we can omit the sub-node. This saves time and space in constructing the reachability tree for deciding NCCA.

\section{Analysis of NCCAs}
\label{property}

Based on the theories developed in previous section, we can easily decide whether a given CA of size $n$ is an NCCA or not. Following are the steps of such a decision procedure.

\begin{enumerate}
\item Form the root of reachability tree. Assign weights to RMTs at root as 0.\\\vspace{-6mm}
\item Get the nodes of the next level and find weights of RMTs at each node.\\\vspace{-6mm}
\item If Theorem~\ref{basic1} or Corollary~\ref{basic2} is disobeyed, output `No' and stop.\\\vspace{-6mm}
\item Remove the sub-nodes, if any, and proceed to the next level with only remaining nodes.\\\vspace{-6mm}
\item Repeat steps 2 to 4 until the leaves are received.\\\vspace{-6mm}
\item If there is any RMT at (leaf) node with non-zero weight, output `No'; otherwise, output `Yes'.
\end{enumerate}

The above procedure can successfully decide a finite CA as NCCA. In this procedure, reachability tree for the given CA is constructed. The number of nodes in a reachability tree grows exponentially. However, the scenario improves if we get a good number of sub-nodes in a level. Practically, number of nodes excluding their sub-nodes at any level is limited, and the upper limit is $4^4 = 256$. Hence, the tree does not grow exponentially in the procedure, which implies it is efficient.

The decision procedure can further be improved if we find the weights of individual RMTs at nodes. In the theories of previous section, we have used weights, but have not found out possible weights of RMTs. If we can get possible weights of an RMT $r \in \Gamma_k^{N_{i.j}}$ and if we observe that $W_k(i, r)$ is not one of the possible weights, we will be able to conclude that the CA is not an NCCA. In this case, we need not to wait upto last level to get negative answer.

\subsection{Weights of RMTs}
\label{weight_RMT}

To find the possible weights of RMTs of rules that participate in NCCA, Lemma~\ref{universal} is instrumental. For each rule ${\mathcal R}_{i}$ of an NCCA, ${\mathcal R}_{i}[0]$ = 0 and ${\mathcal R}_{i}[7]$ = 1. Obviously, $W_0(i, 0) = W_3(i, 7)$ = 0 for each $i \in \{0, 1, \cdots, n\}$. Since sibling RMTs have same weights (see definition of weight in Section~\ref{RTncca}) in any node, so $W_0(i, 1) = W_3(i, 6)$ = 0.

Using these weights, we will be able to find the possible weights of other RMTs. Please observe that all of the 8 RMTs can be present in $\Gamma_k^{N_{i.j}}$ when $i \geq 2$. From the definition of weight, we understand that the weights of RMTs of different levels are related. For example, $W_k(i+1, 0)$ = $W_k(i, 4)$ -1 if ${\mathcal R}_{i}[4]$ = 1, and $W_k(i+1, 0)$ = $W_k(i, 4)$ if ${\mathcal R}_{i}[4]$ = 0. Now, since $W_0(i+1, 0)$ = 0, the weight of RMT 4 $\in \Gamma_0^{N_{i.j}}$ has to be either 0 or 1. Similarly, the weights of RMTs 2, 3 $\in \Gamma_0^{N_{i.j}}$ can either be -1 or 0 because $W_0(i-1, 1)$ = 0 and $W_0(i, 2)$ = $W_0(i-1, 1)$-1 if ${\mathcal R}_{i-1}$[1] = 1 but $W_0(i, 2)$ = $W_0(i-1, 1)$ if ${\mathcal R}_{i-1}$[1] = 0. Now, since possible weights of RMT 3 is -1 or 0, possible weights for RMTs 6 and 7 are -1, 0 and 1. Hence we get possible weights of all RMTs that belong to $\Gamma_0^{N_{i.j}}$, considering $W_0(i, 0) = W_0(i, 1)$ = 0. One can similarly find the possible weights of all RMTs that belong to $\Gamma_3^{N_{i.j}}$ considering $W_3(i, 6) = W_3(i, 7)$ = 0. These weights are noted in Tab.~\ref{possible_weight}.

We can also find the weights of RMTs that are in $\Gamma_1^{N_{i.j}}$ and  $\Gamma_2^{N_{i.j}}$. In case of $\Gamma_1^{N_{i.j}}$, only RMTs 2 and 3 exists when $i = n$. So, for NCCA, $W_1(n, 2) = W_1(n, 3)$ = 0. Again, $W_1(0, 2) = W_1(0, 3)$ = 0. Hence, $W_1(1, 4)$ can be either 0 or 1 depending on ${\mathcal R}_{0}$. Now, if 
$W_1(1, 4)$ = 0 and ${\mathcal R}_{1}$[4] = 0, then $W_1(2, 0) = W_1(2, 1)$ = 0; and if $W_1(1, 4)$ = 0 but ${\mathcal R}_{1}$[4] = 1, then $W_1(2, 0) = W_1(2, 1)$ = -1. On the other hand, if $W_1(1, 4)$ = 1 then weights of RMTs 0, 1 $\in \Gamma_1^{N_{2.j}}$ can be 0 or 1. From the above logic, we get the possible weights of RMTs 0, 1 as -1, 0 and 1. However, once the RMT 0 gets a weight, that weight is carried forward to the lower levels, because in NCCA ${\mathcal R}_{i}[0]$ = 0 for any rule ${\mathcal R}_{i}$. Hence, $W_1(n-1, 0) = W_1(n-1, 1)$ = -1 when $W_1(2, 0)$ = -1. If this happens, then $W_1(n, 2)$ and $W_1(n, 3)$ can never be 0 for any rule ${\mathcal R}_{n-1}$. This implies, -1 can not be a possible weight for RMTs 0, 1 $\in \Gamma_1^{N_{i.j}}$ for any $i$. Hence, possible weights of RMTs 0, 1 $\in \Gamma_1^{N_{i.j}}$ are 0 and 1. With the similar logic we can get that possible weights of RMTs 6, 7 $\in \Gamma_1^{N_{i.j}}$ are 0 and 1. 

We can further get the possible weights of RMTs 2, 3 $\in \Gamma_1^{N_{i.j}}$ apart from weight 0. The possible weights of RMTs 0, 1 $\in \Gamma_1^{N_{i.j}}$ are 0 and 1. Now, $W_1(i, 2)$ (or $W_1(i, 3)$) and $W_1(i-1, 1)$ are related. Depending on ${\mathcal R}_{i-1}$, weights of RMTs 2, 3 $\in \Gamma_1^{N_{i.j}}$ can be -1, 0 and 1 (see the definition of weights in Section~\ref{RTncca}). The possible weights of RMTs 6, 7 $\in \Gamma_1^{N_{i.j}}$ are 0 and 1.  On the other hand,  $W_1(i, 4)$/$W_1(i, 5)$ is dependent on $W_1(i-1, 6)$ and ${\mathcal R}_{i-1}$. Hence, possible weights of RMTs 4, 5 $\in \Gamma_1^{N_{i.j}}$ are 0, 1 and 2. Hence, we get possible weights of all the RMTs that may exist in $\Gamma_1^{N_{i.j}}$. Following above structures of argument, one can also find the possible weights of RMTs that may exist in $\Gamma_2^{N_{i.j}}$.

Tab.~\ref{possible_weight} notes the possible weights of all possible RMTs of a node. First column shows the RMTs of a rule, whereas next four columns state the possible weights of RMTs that are in four consecutive sets of a node.

\begin{table}
\caption{Possible Weights of RMTs}
\begin{center}
\label{possible_weight}
\begin{tabular}{|c|c|c|c|c|}
\hline 
\multirow{2}{*}{RMTs of ${\mathcal R}_{i}$} & \multicolumn{4}{c|}{Possible Weights of RMTs in} \\ \cline{2-5}
 & $\Gamma_0^{N_{i.j}}$ & $\Gamma_1^{N_{i.j}}$ & $\Gamma_2^{N_{i.j}}$ & $\Gamma_3^{N_{i.j}}$ \\ 
\hline 
0 & 0 & 0, 1 & -1, 0 & -1, 0, 1 \\ 
\hline 
1 & 0 & 0, 1 & -1, 0 & -1, 0, 1 \\ 
\hline 
2 & -1, 0 & -1, 0, 1 & -2, -1, 0 & -1, 0 \\ 
\hline 
3 & -1, 0 & -1, 0, 1 & -2, -1, 0 & -1, 0 \\ 
\hline 
4 & 0, 1 & 0, 1, 2 & -1, 0, 1 & 0, 1  \\ 
\hline
5 & 0, 1 & 0, 1, 2 & -1, 0, 1 & 0, 1  \\  
\hline
6 & -1, 0, 1 & 0, 1 & -1, 0 & 0 \\ 
\hline
7 & -1, 0, 1 & 0, 1 & -1, 0 & 0 \\ 
\hline
\end{tabular} 
\end{center}
\end{table}

As mentioned before, possible weights of RMTs can improve our decision algorithm. As soon as we get a weight of an RMT which is not consistent with Tab.~\ref{possible_weight}, we declare the CA as not an NCCA. However, we can get little more improvement. Observe again that, weight of an RMT $r$ (i.e., $W_k(i+1, r)$) depends on weight of another RMT $s$ ($W_k(i, s)$) and ${\mathcal R}_{i}[s]$ where $r$ = $2s$ or $2s+1 \pmod{8}$. Now, when $W_k(i, s)$ is consistent with Tab.~\ref{possible_weight}, then by observing ${\mathcal R}_{i}[s]$ one can declare whether $W_k(i+1, r)$ is going to be consistent or not. If not consistent, we can declare the CA as not an NCCA. It help us decide negatively before observing an inconsistent weight. In case of NCCA, the weights are always consistent. Following are the values of an RMT $r \in \Gamma_k^{N_{i.j}}$ (that is, ${\mathcal R}_{i}[r]$) when $W_k(i, r)$ attains a specific value.\\

\textbf{For $\Gamma_0^{N_{i.j}}$:}
\begin{enumerate}
\item If $W_0(i, 2)$ = -1 then ${\mathcal R}_{i}[2]$ = 0.
\item If $W_0(i, 4)$ = 0 then ${\mathcal R}_{i}[4]$ = 0. If $W_0(i, 4)$ = 1 then ${\mathcal R}_{i}[4]$ = 1.
\item If $W_0(i, 5)$ = 1 then ${\mathcal R}_{i}[5]$ = 1.
\item If $W_0(i, 6)$ = -1 then ${\mathcal R}_{i}[6]$ = 0. If $W_0(i, 6)$ = 1 then ${\mathcal R}_{i}[6]$ = 1.
\end{enumerate}

\textbf{For $\Gamma_1^{N_{i.j}}$:}

\begin{enumerate}[resume]
\item If $W_1(i, 2)$ = -1 then ${\mathcal R}_{i}[2]$ = 0.
\item If $W_1(i, 3)$ = -1 then ${\mathcal R}_{i}[3]$ = 0. If $W_1(i, 3)$ = 1 then ${\mathcal R}_{i}[3]$ = 1.
\item If $W_1(i, 4)$ = 0 then ${\mathcal R}_{i}[4]$ = 0. If $W_1(i, 4)$ = 2 then ${\mathcal R}_{i}[4]$ = 1.
\item If $W_1(i, 5)$ = 2 then ${\mathcal R}_{i}[5]$ = 1.
\end{enumerate}

\textbf{For $\Gamma_2^{N_{i.j}}$:}

\begin{enumerate}[resume]
\item If $W_2(i, 2)$ = -2 then ${\mathcal R}_{i}[2]$ = 0.
\item If $W_2(i, 3)$ = -2 then ${\mathcal R}_{i}[3]$ = 0. If $W_2(i, 3)$ = 0 then ${\mathcal R}_{i}[3]$ = 1.
\item If $W_2(i, 4)$ = -1 then ${\mathcal R}_{i}[4]$ = 0. If $W_2(i, 4)$ = 1 then ${\mathcal R}_{i}[4]$ = 1.
\item If $W_2(i, 5)$ = 1 then ${\mathcal R}_{i}[5]$ = 1.
\end{enumerate}

\textbf{For $\Gamma_3^{N_{i.j}}$:}

\begin{enumerate}[resume]
\item If $W_3(i, 1)$ = -1 then ${\mathcal R}_{i}[1]$ = 0. If $W_3(i, 1)$ = 1 then ${\mathcal R}_{i}[1]$ = 1.
\item If $W_3(i, 2)$ = -1 then ${\mathcal R}_{i}[2]$ = 0.
\item If $W_3(i, 3)$ = -1 then ${\mathcal R}_{i}[3]$ = 0. If $W_3(i, 3)$ = 0 then ${\mathcal R}_{i}[3]$ = 1.
\item If $W_3(i, 5)$ = 1 then ${\mathcal R}_{i}[5]$ = 1.
\end{enumerate}

In our proposed decision algorithm, we use these conditions to decide an NCCA. Apart from the above relations, however, we can also derive relations between weights of different RMTs, exploring the results of previous section. Following are two obvious results. These are also needed in the decision algorithm.

\begin{corollary}
\label{coro_40} 
For any level $i$, and any $k \in \{0, 1, 2, 3\}$: 
\begin{enumerate}
\vspace{-1mm}
\item $W_k(i, 4)$ $\geq$ $W_k(i, 0)$. If $W_k(i, 4)$ = $W_k(i, 0)$ then ${\mathcal R}_i[4]$ = 0. If $W_k(i, 4)$ $>$ $W_k(i, 0)$ then ${\mathcal R}_i[4]$ = 1. \vspace{-2mm}
\item $W_k(i, 5)$ $\geq$ $W_k(i, 1)$. If $W_k(i, 5)$ $>$ $W_k(i, 1)$ then ${\mathcal R}_i[5]$ = 1 and ${\mathcal R}_i[1]$ = 0. 
\end{enumerate} 
\end{corollary}

\begin{proof}
{\em Case 1:} From the definition of weight, we get $W_k(i+1, 0)$ is either $W_k(i, 4)$ - 1 if ${\mathcal R}_i[4]$ = 1, or $W_k(i, 4)$ if ${\mathcal R}_i[4]$ = 0. Further, $W_k(i+1, 0)$ = $W_k(i, 0)$ as ${\mathcal R}_i[0]$ = 0 for an NCCA. Hence, $W_k(i, 4)$ $\geq$ $W_k(i, 0)$ for an NCCA. And, when $W_k(i, 4)$ = $W_k(i, 0)$
then ${\mathcal R}_i[4]$ = 0, and when $W_k(i, 4)$ $>$ $W_k(i, 0)$ then ${\mathcal R}_i[4]$ = 1.

{\em Case 2:} As RMT 1 and RMT 5 are the sibling of RMT 0 and RMT 4 respectively, then at any level $i$, 
$W_k(i, 5)$ $\geq$ $W_k(i, 1)$ for each $k$. Further, when $W_k(i, 5)$ $>$ $W_k(i, 1)$ then ${\mathcal R}_i[5]$ = 1 and ${\mathcal R}_i[1]$ = 0.
\end{proof}

\begin{corollary} 
\label{coro_37}
For any level $i$, and any $k \in \{0, 1, 2, 3\}$: 
\begin{enumerate}
\vspace{-1mm}
\item $W_k(i, 3)$ $\leq$ $W_k(i, 7)$. If $W_k(i, 3)$ = $W_k(i, 7)$ then ${\mathcal R}_i[3]$ = 1. If $W_k(i, 3)$ $<$ $W_k(i, 7)$ then ${\mathcal R}_i[3]$ = 0. \vspace{-2mm}
\item $W_k(i, 2)$ $\leq$ $W_k(i, 6)$. If $W_k(i, 2)$ $<$ $W_k(i, 6)$ the ${\mathcal R}_i[2]$ = 0 and ${\mathcal R}_i[6]$ = 1. 
\end{enumerate}
\end{corollary}

\begin{proof}
{\em Case 1:} According to the definition of weight, $W_k(i+1, 7)$ = $W_k(i, 3)$ if ${\mathcal R}_i[3]$ = 1; $W_k(i+1, 7)$ = $W_k(i, 3)$ + 1 if ${\mathcal R}_i[3]$ = 0. Also, $W_k(i+1, 7)$ = $W_k(i, 7)$ as ${\mathcal R}_i[7]$ = 1 for an NCCA. Hence, $W_k(i, 3)$ $\leq$ $W_k(i, 7)$. And, when $W_k(i, 3)$ $<$ $W_k(i, 7)$, ${\mathcal R}_i[3]$ = 0; otherwise ${\mathcal R}_i[3]$ = 1.

{\em Case 2:} As RMT 2 and RMT 6 are the sibling of RMT 3 and RMT 7 respectively, then at any level $i$, 
$W_k(i, 2)$ $\leq$ $W_k(i, 6)$ for each $k$. Further, ${\mathcal R}_i[2]$ = 0 and ${\mathcal R}_i[6]$ = 1 when $W_k(i, 2)$ $<$ $W_k(i, 6)$.
\end{proof}

From all the above discussions, it is clear that an arbitrary rule cannot take part in a rule vector of NCCA. For example, if next state value of RMTs 0  and 7 of a rule are not 0 and 1 respectively, the rule can not participate in an NCCA (Lemma~\ref{universal}). Based on this, we classify the rules as {\em number conserving} and {\em non-number conserving} rules. Next we identify the number conserving rules.

\subsection{Number conserving rules}
\label{NC_rule}

\begin{definition}
A rule is \textbf{non-number conserving rule} if its presence in a rule vector makes the CA non-NCCA. Otherwise, it is a \textbf{number conserving rule}.
\end{definition}

\begin{example}
The 5-cell CA with rule vector $\langle 170, 240, 238, 192, 204 \rangle$ is a NCCA. Therefore, all of the five rules are number conserving rules. On the other hand, a CA with rule vector $\langle 170, 240, 239, 192, 204 \rangle$ is a non-NCCA. The rule 239 makes the CA non-NCCA. So, rule 239 is a {\em non-number conserving rule}.
\end{example}

\begin{theorem}
\label{rule_select}
A rule ${R}$ is number conserving rule if the following conditions are satisfied--

\begin{enumerate}
\item \label{rule_select_c1} ${R}[0]$ = 0 and ${R}[7]$ = 1.
\item \label{rule_select_c2} ${R}[0]$ = ${R}[4]$ and ${R}[1]$ = ${R}[5]$, \textbf{or} ${R}[0]$ = ${R}[1]$ and ${R}[4]$ = ${R}[5]$.
\item \label{rule_select_c3} ${R}[2]$ = ${R}[6]$ and ${R}[3]$ = ${R}[7]$, \textbf{or} ${R}[2]$ = ${R}[3]$ and ${R}[6]$ = ${R}[7]$.
\end{enumerate}
\end{theorem}

\begin{proof}
{\em Case 1:} We get the condition 1 directly from Lemma~\ref{universal}.

To prove Case~\ref{rule_select_c2} and Case~\ref{rule_select_c3}, let us consider a CA including $R$ as ${\mathcal R}_i$ -- $i^{th}$ rule of a rule vector of size $n$, where $i \geq 2$ and $n \geq 5$. Therefore, all the 8 RMTs are present in $\cup_j\Gamma_k^{N_{i.j}}$ for any value of $k \in \{0, 1, 2, 3\}$. Without loss of generality we further consider that all the properties of NCCA, discussed till now are maintained for rules upto ${\mathcal R}_{i-1}$. Here, we prove Case~\ref{rule_select_c2} and Case~\ref{rule_select_c3} by method of contradiction. That is, we consider ${\mathcal R}_i$ does not obey the conditions of Case~\ref{rule_select_c2} and Case~\ref{rule_select_c3}. Then we show that the CA can never be an NCCA. 

{\em Case 2:} First we consider that ${\mathcal R}_i[0]$ = ${\mathcal R}_i[4]$ and ${\mathcal R}_i[1] \neq {\mathcal R}_i[5]$. Now according to the Corollary~\ref{1110}, $W_k(i, 1) \neq W_k(i, 5)$. Further, RMT 5 is the sibling of RMT 4 and RMT 1 is the sibling of RMT 0. That is, $W_k(i, 4) = W_k(i, 5)$ and $W_k(i, 0) = W_k(i, 1)$ (by definition of weight). So, $W_k(i, 0) \neq W_k(i, 4)$, which is not possible when ${\mathcal R}_i[0]$ = ${\mathcal R}_i[4]$ (Corollary~\ref{equivalentTh1}). Hence, ${\mathcal R}_i$ can not be a number conserving rule when ${\mathcal R}_i[0]$ = ${\mathcal R}_i[4]$ and ${\mathcal R}_i[1] \neq {\mathcal R}_i[5]$. In this way we can also prove that ${\mathcal R}_i[1]$ = ${\mathcal R}_i[5]$ but ${\mathcal R}_i[0] \neq {\mathcal R}_i[4]$ is not possible. 

Next we consider that ${\mathcal R}_i[0]$ = ${\mathcal R}_i[1]$ and ${\mathcal R}_i[4] \neq {\mathcal R}_i[5]$. That is, if ${\mathcal R}_i[4]$ = 0 then ${\mathcal R}_i[5]$ = 1. According to the Lemma~\ref{universal}, ${\mathcal R}_i[0]$ = 0, so ${\mathcal R}_i[1]$ = 0. Now, ${\mathcal R}_i[0]$ = 0 and ${\mathcal R}_i[4]$ = 0 implies $W_k(i, 0)$ = $W_k(i, 4)$ (Corollary~\ref{equivalentTh1}). Further, RMT 1 is the sibling of RMT 0 and RMT 5 is the sibling of RMT 4, so $W_k(i, 1)$ = $W_k(i, 5)$ since $W_k(i, 4) = W_k(i, 5)$ and $W_k(i, 0) = W_k(i, 1)$. But it violates Corollary~\ref{1110} as ${\mathcal R}_i[1]$ = 0 and ${\mathcal R}_i[5]$ = 1. So, ${\mathcal R}_i[0]$ = ${\mathcal R}_i[1]$ and ${\mathcal R}_i[4] \neq {\mathcal R}_i[5]$ is not possible. One can also assume that ${\mathcal R}_i[4]$ = 1 and ${\mathcal R}_i[5]$ = 0. It can easily be shown that this is also not possible. Hence, ${\mathcal R}_i$ can not be a number conserving rule when ${\mathcal R}_i[0]$ = ${\mathcal R}_i[1]$ and ${\mathcal R}_i[4] \neq {\mathcal R}_i[5]$. Similarly, we can prove that ${\mathcal R}_i[4]$ = ${\mathcal R}_i[5]$ but ${\mathcal R}_i[0] \neq {\mathcal R}_i[1]$ is not possible.

Finally consider that ${\mathcal R}_i[0] \neq {\mathcal R}_i[4]$, ${\mathcal R}_i[1] \neq {\mathcal R}_i[5]$, ${\mathcal R}_i[0] \neq {\mathcal R}_i[1]$ and ${\mathcal R}_i[4] \neq {\mathcal R}_i[5]$. That means ${\mathcal R}_i[0]$ = ${\mathcal R}_i[5]$ and ${\mathcal R}_i[1]$ = ${\mathcal R}_i[4]$. Since ${\mathcal R}_i[0] \neq {\mathcal R}_i[4]$ and ${\mathcal R}_i[1] \neq {\mathcal R}_i[5]$, according to Corollary~\ref{1110}, $W_k(i, 0) \neq W_k(i, 4)$ and $W_k(i, 1) \neq W_k(i, 5)$. Now, according to Corollary~\ref{coro_40}, $W_k(i, 5) \geq W_k(i, 1)$. In this case, $W_k(i, 5) > W_k(i, 1)$. Now, ${\mathcal R}_i[0]$ = ${\mathcal R}_i[5]$ and ${\mathcal R}_i[0] \neq {\mathcal R}_i[1]$ means ${\mathcal R}_i[5]$ = 0 and ${\mathcal R}_i[1]$ = 1. RMT 1 and RMT 5 are equivalent RMTs, so both of them contribute same set of RMTs, but the weight of the RMTs in both the cases are different because $W_k(i, 5) > W_k(i, 1)$, ${\mathcal R}_i[5]$ = 0 and ${\mathcal R}_i[1]$ = 1. According to the Theorem~\ref{basic1}, a RMT present in $\Gamma_k^{N_{i.j}}$ with different weight is not possible. So, our consideration is false.

Combining all the above, we conclude that a number conserving rule has to obey the condition of Case~\ref{rule_select_c2}.\\
{\em Case 3:} The proof for the RMTs 2, 3, 6 and 7 is identical upto the 0/1 exchange of RMTs 0, 1, 4 and 5.\\

However, the conditions of Case~\ref{rule_select_c2} and Case~\ref{rule_select_c3} are also applicable to first rule (${\mathcal R}_0$) and second rule (${\mathcal R}_1$). Under periodic boundary condition any rule can be considered as first rule and this consideration does not alter the behavior of CA, if the sequence of the rules remains same. That is, the dynamic behavior of two CAs -- $\langle {\mathcal R}_{0}, {\mathcal R}_{1}, \cdots, {\mathcal R}_{i}, \cdots, {\mathcal R}_{n-1} \rangle$ and $\langle {\mathcal R}_{i}, {\mathcal R}_{i+1},\cdots, {\mathcal R}_{n-1}, {\mathcal R}_{0}, {\mathcal R}_{1}, \cdots, {\mathcal R}_{i-1} \rangle$ are same. Hence the proof.
\end{proof} 

There are nine number conserving rules that respect the conditions of Theorem~\ref{rule_select}. The rules are 136, 170, 184, 192, 204, 226, 238, 240, 252. These rules can form a rule vector of a NCCA of size $n$. However, for very small values of $n$, we can get some additional rules as number conserving rules. Following corollary states this fact.

\begin{corollary} 
\label{Th2_true}
The rules that obey the conditions of Theorem~\ref{rule_select} are the only number conserving rules which can form a rule vector of size $n \geq 5$.
\end{corollary} 

\begin{proof}
In the root, $\Gamma_0^{N_{0.0}}$ = \{0, 1\}, $\Gamma_1^{N_{0.0}}$ = \{2, 3\}, $\Gamma_2^{N_{0.0}}$ = \{4, 5\} and $\Gamma_3^{N_{0.0}}$ = \{6, 7\}. In level 1, $\Gamma_k^{N_{1.0}}$ $\cup$ $\Gamma_k^{N_{1.1}}$ is either \{0, 1, 2, 3\} or \{4, 5, 6, 7\}. In level 2, however, $\cup_j\Gamma_k^{N_{2.j}}$ is \{0, 1, 2, 3, 4, 5, 6, 7\}. To prove the effectiveness of conditions of Case 2 and Case 3 of Theorem~\ref{rule_select}, we need to get all the 8 RMTs. Since $\cup_j\Gamma_k^{N_{n-2.j}}$ and $\cup_j\Gamma_k^{N_{n-1.j}}$ do not contain all the 8 RMTs (Point~\ref{rtd5} and Point~\ref{rtd6} of Definition~\ref{Rtree_def}), we can get all the RMTs at least in one level if $n-3 \geq 2$. This implies $n \geq 5$. Hence proved.
\end{proof}

For example, if $n = 4$, then we get additional 6 number conserving rules -- 160, 172, 202, 216, 228 and 250. Obviously, these rules are non-number conserving rules when $n \geq 5$.

However, an arbitrary arrangement of number conserving rules do not form a rule vector of an NCCA. As example, let us consider two rule vectors-- ${\mathcal R}$ = $\langle 192, 136, 184, 252, 204, 238 \rangle$ and ${\mathcal R}^{''}$ = $\langle 252, 204, 192,\\ 136, 184, 238 \rangle$. All the rules of rule vectors are number conserving rules. Now, ${\mathcal R}$ is an NCCA but ${\mathcal R}^{''}$ is not. This implies a specific sequence of number conserving rules forms a NCCA.

\begin{lemma}
\label{sequenceOfRule}
Only a specific sequence of number conserving rules forms a NCCA.
\end{lemma}


\section{The Decision Algorithm}
\label{algorithms}

In this section, we present a decision algorithm (Algorithm~\ref{analysisNCCA}), which decides whether a given rule vector is NCCA or not. To do this, the algorithm uses the theories developed in previous sections. However, the theories of previous sections lead to the following observations :\\
(i) weights of an RMT $r$ of ${\mathcal R}_{i}$ that belongs to $\Gamma_k^{N_{i.j_1}}$, $\Gamma_k^{N_{i.j_2}}$, $\cdots$ are same,\\ (ii) behavior of weight of an RMT does not depend on node number of a level, but depends on the set number in a node $N_{i.j}$ -- whether it is $\Gamma_0^{N_{i.j}}$, $\Gamma_1^{N_{i.j}}$, $\Gamma_2^{N_{i.j}}$ or $\Gamma_3^{N_{i.j}}$, and\\ (iii) final decision is taken based only on the weights of RMTs at leaves.

Therefore, what only matters is the weight of an RMT that belongs to $\Gamma_k^{N_{i.j}}$ for a given $k \in \{0, 1, 2, 3\}$.

Now consider, $\Gamma_k^i$ = $\cup_j\Gamma_k^{N_{i.j}}$. By definition, weight of an RMT $r \in \Gamma_k^0$ is 0. Then, we can get weights of all RMTs of $\Gamma_k^1$ following ${\mathcal R}_{0}$ and the definition of weight. In this way, we can get the weights of all RMTs of $\Gamma_k^i$ for any $i$ ($0 \leq i \leq n$) after knowing ${\mathcal R}_{i-1}$. During the weight calculation, if we observe that the weight of an RMT is not consistent with Tab.~\ref{possible_weight}, we conclude that the CA is not an NCCA. It can be noticed here that, $\Gamma_k^i$ contains all the 8 RMTs for all $i \geq 2$ (but $\leq n-3$).

This discussion implies that, to decide a CA as NCCA, we need not to develop all the nodes of a tree. Rather it is sufficient to work only with $\Gamma_0^i$, $\Gamma_1^i$, $\Gamma_2^i$ and $\Gamma_3^i$ for any level $i$. Here, we can think of a {\em super node} ${\mathcal N}_i$ = ($\Gamma_0^i$, $\Gamma_1^i$, $\Gamma_2^i$, $\Gamma_3^i$). Please note that any node $N_{i.j}$ of a level $i$ is sub-node of ${\mathcal N}_i$.

The proposed algorithm develops the {\em super node} of a level $i$, and finds the weights of RMTs of the node. The algorithm uses two data structures, $\Gamma_k$ -- to store the RMTs present in $\Gamma_k^i$ of level $i$, and $W_k$ -- to store the weight of RMTs in $\Gamma_k$ for each $k \in \{0, 1, 2, 3\}$. We do not use level number in the data structures, because the algorithm deals with RMTs of only one level and their weights. The algorithm performs two tasks -- (i) finding of weights of RMTs in a super node of a level, and (ii) verifying whether the weights are consistent with Tab~\ref{possible_weight}. If weight of any RMT is inconsistent, the algorithm stops with negative answer, otherwise it continues in finding of weights of RMTs and verification of inconsistency of weights, if any.

To find the weights of RMTs, we develop a procedure, named {\em FindNextWeight()}. As argument, the procedure takes a rule ${R}$, a set of RMTs $\Gamma$ (= $\Gamma_k^i$) and $W$ (that is, $W_k(i, r)$ for each $r \in \Gamma_k^i$). It uses two temporary data structures -- $Temp\Gamma$ and $TempW$. The procedure finds the RMTs of $\Gamma_k^{i+1}$ and weights of RMTs. These new set of RMTs and their weights are assigned to $\Gamma$ and $W$, respectively. 

\begin{procedure}[hbtp]
	\BlankLine
	
	\Begin{
				Temp$\Gamma \leftarrow \emptyset$ \;
				\ForEach {RMT $r \in \Gamma$}
				{ 
					\If{$r \in \{0, 1, 4, 5\}$}
					{
						\If{${\mathcal R}[r]$=1}
						{
							Temp$W[2r \pmod {8}]$ $\leftarrow$ $W[r] - 1$ \;
							Temp$W[(2r+1) \pmod {8}]$ $\leftarrow$ $W[r] - 1$ \;		
						}
						\Else
						{
							Temp$W[2r \pmod {8}]$ $\leftarrow$ $W[r]$ \;
							Temp$W[(2r+1) \pmod {8}]$ $\leftarrow$ $W[r]$ \;	
						}				
					}
					\Else
					{
						\If{${\mathcal R}[r]$=0}
						{
							Temp$W[2r \pmod {8}]$ $\leftarrow$ $W[r] + 1$ \;
							Temp$W[(2r+1) \pmod {8}]$ $\leftarrow$ $W[r] + 1$ \;	
						}
						\Else
						{
							Temp$W[2r \pmod {8}]$ $\leftarrow$ $W[r]$ \;
							Temp$W[(2r+1) \pmod {8}]$ $\leftarrow$ $W[r]$ \;	
						}
					}
						Temp$\Gamma$ $\leftarrow$ Temp$\Gamma$ $\cup$ \{$2r \pmod {8}, (2r+1) \pmod {8}$\} \;				
				}
						Assign $\Gamma$ $\leftarrow$ Temp$\Gamma$ \; 
						~~~~~~~~~~~~$W \leftarrow$ Temp$W$ \;
			}		
	\caption{FindNextWeight(${R}$, $\Gamma$, $W$)}
	\label{weight_procedure}
\end{procedure}

After getting the weights, the proposed algorithm verifies whether the weights of RMTs, are consistent with Tab.~\ref{possible_weight}. To do so, the algorithm implements the conditions relating $W_k(i, r)$ and ${\mathcal R}_{i}$ (see Section~\ref{weight_RMT}). Further, it uses Corollary~\ref{coro_40} and Corollary~\ref{coro_37} to do the same.

The descriptions of the steps of the algorithm are noted in Algorithm~\ref{analysisNCCA}. The algorithm takes an $n$-cell CA, and outputs `Yes'' if the CA is an NCCA; ``No'' otherwise.

\begin{algo}
\caption{Decide if a given CA is NCCA}
\label{analysisNCCA}
{\bf Input:} CA with rule vector ${\mathcal{R}}= \langle{\mathcal{R}}_0, {\mathcal{R}}_1, \cdots, {\mathcal{R}}_{n-1}\rangle$\\
{\bf Output:} `Yes' if the CA is NCCA; `No' otherwise.\\
\rule[4pt]{0.95\textwidth}{0.95pt}\\

\textbf{Step 1 :}  Set $\Gamma_0 \leftarrow \{0, 1\}$, $\Gamma_1 \leftarrow \{2, 3\}$, $\Gamma_2 \leftarrow \{4, 5\}$, $\Gamma_3 \leftarrow \{6, 7\}$,\\~~~~~~~~~~~~~ $W_k[r] \leftarrow 0$ if $r \in \Gamma_k$, for each $k$, $0 \leq k \leq 3$.\\
\textbf{Step 2 :}  If ${\mathcal{R}}_0 \notin$ \{136, 170, 184, 192, 204, 226, 238, 240, 252\} then report `No' and Return.\\
~~~~~~~~~~~~~~Otherwise call {\it FindNextWeight(${\mathcal R}_{0}$, $\Gamma_k$, $W_k$)} for each $k \in \{0, 1, 2, 3\}$. \\
\textbf{Step 3 :}  Set $i \leftarrow 1 $.\\
\textbf{Step 4 :}  If ${\mathcal{R}}_i \notin$ \{136, 170, 184, 192, 204, 226, 238, 240, 252\} then report `No' and Return.\\
\textbf{Step 5 :}  If any of the following conditions is not satisfied then report `No' and Return.\\
~~~~~~~~~~~~~~(i) $W_0[2]$ = -1 $\implies$ ${\mathcal{R}}_i[2]$  = 0 ;\\ 
~~~~~~~~~~~~~~(ii) ($W_0[4]$ = 0 $\implies$ ${\mathcal{R}}_i[4]$  = 0) $\wedge$ ($W_0[4]$ = 1 $\implies$ ${\mathcal{R}}_i[4]$  = 1);\\
~~~~~~~~~~~~~~(iii) $W_0[5]$ = 1 $\implies$ ${\mathcal{R}}_i[5]$  = 1 ;\\
~~~~~~~~~~~~~~(iv) ($W_0[6]$ = -1 $\implies$ ${\mathcal{R}}_i[6]$  = 0) $\wedge$ ($W_0[6]$ = 1 $\implies$ ${\mathcal{R}}_i[6]$  = 1);\\
~~~~~~~~~~~~~~(v) $W_1[2]$ = -1 $\implies$ ${\mathcal{R}}_i[2]$  = 0 ;\\
~~~~~~~~~~~~~~(vi) ($W_1[3]$ = -1 $\implies$ ${\mathcal{R}}_i[3]$  = 0) $\wedge$ ($W_1[3]$ = 1 $\implies$ ${\mathcal{R}}_i[3]$  = 1);\\
~~~~~~~~~~~~~~(vii) ($W_1[4]$ = 0 $\implies$ ${\mathcal{R}}_i[4]$  = 0) $\wedge$ ($W_1[4]$ = 2 $\implies$ ${\mathcal{R}}_i[4]$  = 1);\\
~~~~~~~~~~~~~~(viii) $W_1[5]$ = 2 $\implies$ ${\mathcal{R}}_i[5]$  = 1 ;\\
~~~~~~~~~~~~~~(ix) $W_2[2]$ = -2 $\implies$ ${\mathcal{R}}_i[2]$  = 0 ;\\
~~~~~~~~~~~~~~(x) ($W_2[3]$ = -2 $\implies$ ${\mathcal{R}}_i[3]$  = 0) $\wedge$ ($W_2[3]$ = 0 $\implies$ ${\mathcal{R}}_i[3]$  = 1);\\
~~~~~~~~~~~~~~(xi) ($W_2[4]$ = -1 $\implies$ ${\mathcal{R}}_i[4]$  = 0) $\wedge$ ($W_2[4]$ = 1 $\implies$ ${\mathcal{R}}_i[4]$  = 1);\\
~~~~~~~~~~~~~~(xii) $W_2[5]$ = 1 $\implies$ ${\mathcal{R}}_i[5]$  = 1 ;\\
~~~~~~~~~~~~~~(xiii) ($W_3[1]$ = -1 $\implies$ ${\mathcal{R}}_i[1]$  = 0) $\wedge$ ($W_3[1]$ = 1 $\implies$ ${\mathcal{R}}_i[1]$  = 1);\\
~~~~~~~~~~~~~~(xiv) $W_3[2]$ = -1 $\implies$ ${\mathcal{R}}_i[2]$  = 0 ;\\
~~~~~~~~~~~~~~(xv) ($W_3[3]$ = -1 $\implies$ ${\mathcal{R}}_i[3]$  = 0) $\wedge$ ($W_3[3]$ = 0 $\implies$ ${\mathcal{R}}_i[3]$  = 1);\\
~~~~~~~~~~~~~~(xvi) ($W_3[5]$ = 1 $\implies$ ${\mathcal{R}}_i[5]$  = 1) ;\\
\textbf{Step 6 :}  If any of the following conditions is not satisfied for any $k \in$ \{0, 1, 2, 3\} then report `No' and return\;
~~~~~~~~~~~~~~(i) $W_k[4]$ = $W_k[0]$ $\implies {\mathcal{R}}_i[4]$  = 0; \\
~~~~~~~~~~~~~~(ii) $W_k[4]$ $>$ $W_k[0]$ $\implies {\mathcal{R}}_i[4]$  = 1 ;\\
~~~~~~~~~~~~~~(iii) $W_k[5]$ $>$ $W_k[1]$ $\implies ({\mathcal{R}}_i[5]$  = 1 $\wedge$ ${\mathcal{R}}_i[1]$  = 0) ;\\
~~~~~~~~~~~~~~(iv) $W_k[3]$ = $W_k[7]$ $\implies {\mathcal{R}}_i[3]$  = 1 ;\\
~~~~~~~~~~~~~~(v) $W_k[3]$ $<$ $W_k[7]$ $\implies {\mathcal{R}}_i[3]$  = 0 ;\\
~~~~~~~~~~~~~~(vi) $W_k[2]$ $<$ $W_k[6]$ $\implies ({\mathcal{R}}_i[2]$  = 0 $\wedge$ ${\mathcal{R}}_i[6]$  = 1) ;\\
\textbf{Step 7 :}  Call {\it FindNextWeight(${\mathcal R}_{i}$, $\Gamma_k$, $W_k$)} for each $k \in \{0, 1, 2, 3\}$ ;\\
\textbf{Step 8 :}  $i \leftarrow i+1$; \\
\textbf{Step 9 :}  If ($i < n-2$) then goto Step 4 ;\\
~~~~~~~~~~~~~ If ($i = n-2$) then \\
~~~~~~~~~~~~~ Set $\Gamma_0 \leftarrow \Gamma_0 \cap \{0, 2, 4, 6\}$ ; $\Gamma_1 \leftarrow \Gamma_1 \cap \{0, 2, 4, 6\}$ ;\\
~~~~~~~~~~~~~~~~~~~ $\Gamma_2 \leftarrow \Gamma_2 \cap \{1, 3, 5, 7\}$ ; $\Gamma_3 \leftarrow \Gamma_3 \cap \{1, 3, 5, 7\}$ ; \\
~~~~~~~~~~~~~ goto Step 6;\\
~~~~~~~~~~~~~ If ($i = n-1$) then \\
~~~~~~~~~~~~~ Set $\Gamma_0 \leftarrow \Gamma_0 \cap \{0, 4\}$ ; $\Gamma_1 \leftarrow \Gamma_1 \cap \{1, 5\}$ ;  \\
~~~~~~~~~~~~~~~~~~~ $\Gamma_2 \leftarrow \Gamma_2 \cap \{2, 6\}$ ; $\Gamma_3 \leftarrow \Gamma_3 \cap \{3, 7\}$ ; \\
~~~~~~~~~~~~~ goto Step 6;\\
\textbf{Step 10 :} If $W_k[r] \neq 0 $ for any $r \in \Gamma_k$ and any $k \in \{0, 1, 2, 3\}$ then report `No';\\
~~~~~~~~~~~~~~~ Otherwise, `Yes'; 
\end{algo}

In Step 1, the algorithm forms the root (which is also a super node of level 0) of the tree, and initializes the weights of RMTs. Next, it checks whether the rule ${\mathcal R}_{i}$ is number conserving or not (Step 2 and Step 4). If ${\mathcal R}_{i}$ is not a number conserving rule, the algorithm stops with negative answer. Otherwise, for $i = 0$, the weights of RMTs of ${\mathcal N}_{1}$, the super node of level 1 are found out. For $i > 0$, first some conditions relating weights of RMTs of ${\mathcal N}_{i}$ and ${\mathcal R}_{i}$ are checked (Step 5), and the conditions of Corollary~\ref{coro_40} and \ref{coro_37} are verified (Step 6), then the weights of RMTs of ${\mathcal N}_{i+1}$ are found out. In fact, Step 5 implements the 16 conditions presented in Section~\ref{weight_RMT}. The conditions of Step 5 and Step 6 can be read in the following way:\\
``$W_0[2]$ = -1 $\implies$ ${\mathcal{R}}_i[2]$  = 0" : If $W_0[2]$ = -1 and ${\mathcal{R}}_i[2]$  = 0 then the condition is true. But if $W_0[2]$ = -1 but ${\mathcal{R}}_i[2]$ $\neq$ 0, the condition is false.

If any one of the conditions of Step 5 and 6 is violated then the algorithm stops with output `No'. This procedure repeats when $i \leq n-3$. For $i=n-2$ and $i=n-1$, first do the Step 9 and then Step 6, after that do the Step 7. If any one condition of Step 6 is violated then report `No'. At the end ($i=n$), in each $\Gamma_k$, check the weight of each RMT. Finally if any non-zero weight is found, then report `No', otherwise report `Yes'.\\  

\noindent
\textbf{Complexity:} The time requirement of Algorithm~\ref{analysisNCCA} depends on $n$, the size of CA only. Hence, the worse case time complexity of Algorithm~\ref{analysisNCCA} is $O (n)$.

\begin{theorem}
\label{correct_Analysis}
Algorithm~\ref{analysisNCCA} correctly checks whether a rule vector ${\mathcal{R}}= \langle{\mathcal{R}}_0, {\mathcal{R}}_1, \cdots, {\mathcal{R}}_{n-1}\rangle$ is NCCA or not, where $n \geq 5$.
\end{theorem}

\begin{proof}
Let us consider a rule vector ${\mathcal{R}}= \langle{\mathcal{R}}_0, {\mathcal{R}}_1, \cdots, {\mathcal{R}}_{n-1}\rangle$ as input to Algorithm~\ref{analysisNCCA}. It is followed from Theorem~\ref{rule_select} that if rule $\mathcal{R}_i$ of ${\mathcal{R}}$ is not a number conserving rule, then ${\mathcal{R}}$ is not an NCCA. This is verified by Step 2 and Step 4 of Algorithm~\ref{analysisNCCA}. However, if all the rules of ${\mathcal{R}}$ are number conserving rule, then it does not necessarily imply that the CA is an NCCA (Lemma~\ref{sequenceOfRule}). So, we find the {\em weights} of RMTs of each set after scanning each rule of ${\mathcal{R}}$ from ${\mathcal{R}}_{0}$ (Step 2 and Step 7). It is already identified that in case of NCCA, the weights follow some conditions, which are summarized in Section~\ref{weight_RMT}. Whether these conditions are satisfied are checked at Step 5. Further, conditions given by Corollary~\ref{coro_40} and Corollary~\ref{coro_37} are verified by Step 6. If any condition is not satisfied, the algorithm decides "No". Finally, Algorithm~\ref{analysisNCCA} returns "Yes" if after scanning all rules of ${\mathcal{R}}$ the weight of each RMT of $\Gamma_k$ ($k \in \{0, 1, 2, 3\}$) is 0 (Step 10). Therefore, correctness of the algorithm is verified by the theories, which have been developed before.
\end{proof}

\begin{figure}[h]
\begin{center}
\includegraphics[height=3.8in, width=3.4in]{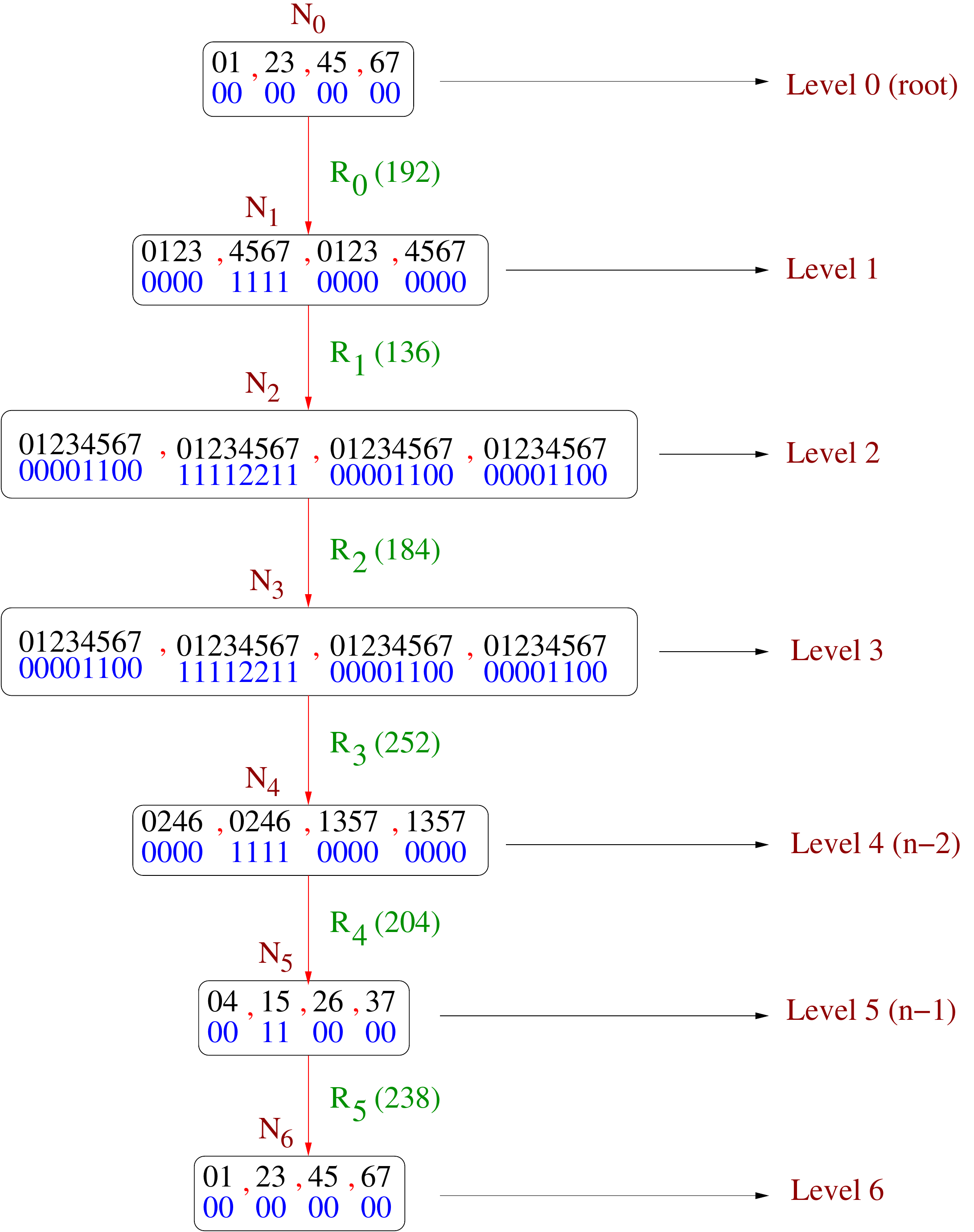}
\caption{Super-node of each level for rule vector $\langle{192, 136, 184, 252, 204, 238}\rangle$}
\label{algo_steps}
\end{center}
\end{figure}

\begin{sidewaysfigure}[hbtp]
 \centering
\includegraphics[height=3.7in,width=6.5in]{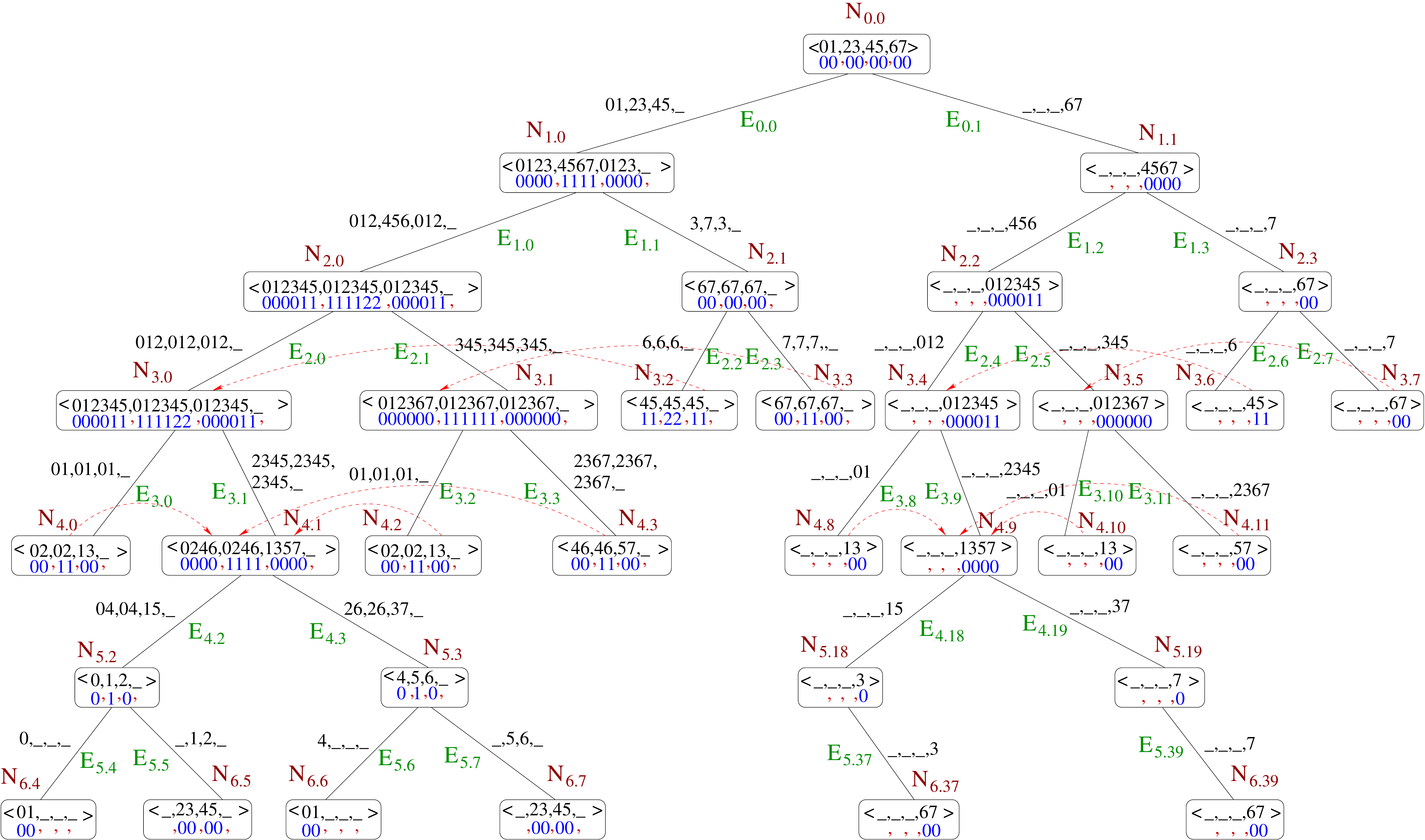}
\caption{The reachability tree of the rule vector $\langle{192, 136, 184, 252, 204, 238}\rangle$ }
\label{NPCAweight}
\end{sidewaysfigure}

\begin{example}
\label{examAnalysis}
Let us consider the CA $\langle{192, 136, 184, 252, 204, 238}\rangle$ as input to Algorithm~\ref{analysisNCCA}. To decide it as NCCA, the algorithm first forms root node ${\mathcal{N}}_{0}$ with four sibling pairs (\{0,1\}, \{2,3\}, \{4,5\}, \{6,7\}), and assigns the weight of each RMT as 0 (Step 1). After that we find the weights of RMTs which are in ${\mathcal{N}}_{1}$ (level 1) by taking ${\mathcal{R}}_{0}$ = 192. The RMTs of level 1 and weight of those RMTs are shown in node ${\mathcal{N}}_{1}$ of Fig.~\ref{algo_steps}. Then we consider the next rule ${\mathcal{R}}_{1}$ = 136, and check whether it is a number conserving rule or not (Step 4). ${\mathcal{R}}_{1}$ is a number conserving rule, so we next check all the conditions of Step 5 and Step 6. We find that the only condition (x) of Step 5 is applicable. In condition (x) of Step 5, ($W_2[3]$ = 0 $\implies$ ${\mathcal{R}}_i[3]$  = 1) means that in $\Gamma_2^1$, the weight of RMT 3 is 0 implies that ${\mathcal{R}}_{1}[3]$ = 1. In rule 136, the next state value of RMT 3 is 1, so the condition (x) of Step 5 is true. In this way, we check Step 6 also. Note that in case of rule 136, no condition of Step 5 and Step 6 are disobeyed. After that in Step 7, we find the weight of RMTs which are in ${\mathcal{N}}_{2}$ (level 2 of Fig.~\ref{algo_steps}). For rule 184 (${\mathcal{R}}_{2}$) and 252 (${\mathcal{R}}_{3}$), we do the similar processing, and observe that no conditions are disobeyed. After that we consider rule 204, which is ${\mathcal{R}}_{n-2}$, so we remove some RMTs (Step 9) and then check the conditions of Step 6. We find that only conditions (i) and (iv) of Step 6 are applicable. From the node ${\mathcal{N}}_{4}$ (level $n-2$) of Fig.~\ref{algo_steps}, we see that in $\Gamma_0^{n-2}$ and $\Gamma_1^{n-2}$, the weights of RMTs 4 and 0 are equal, and ${\mathcal{R}}_{n-2}[4]$ = 0. Similarly, in $\Gamma_2^{n-2}$ and $\Gamma_3^{n-2}$, the weights of RMTs 3 and 7 are equal, and ${\mathcal{R}}_{n-2}[3]$ = 1. Since 204[4] = 0 and 204[3] = 1, the conditions (i) and (iv) of Step 6 are true. Next we find the weights of RMTs of ${\mathcal{N}}_{5}$ (level $n-1$ of Fig.~\ref{algo_steps}). Lastly we consider rule 238, as ${\mathcal{R}}_{n-1}$, and use Step 9. Then we check the conditions of Step 6. From the node ${\mathcal{N}}_{5}$ of Fig.~\ref{algo_steps}, we see that in $\Gamma_0^{n-1}$ the weights of RMTs 4 and 0 are equal, and ${\mathcal{R}}_{n-1}[4]$ = 0. Similarly, in $\Gamma_3^{n-1}$, the weights of RMTs 3 and 7 are equal, and ${\mathcal{R}}_{n-1}[3]$ = 1. Since 238[4] = 0 and 238[3] = 1, the conditions (i) and (iv) of Step 6 are true. Again we calculate the weight of each RMT which is in ${\mathcal{N}}_{6}$ (level 6 of Fig.~\ref{algo_steps}). If we get any non-zero weight for an RMT then return `No'(Step 10). But from node ${\mathcal{N}}_{6}$ (level $n$) of Fig.~\ref{algo_steps}, we clearly see that no RMT has non-zero weight. So, Algorithm~\ref{analysisNCCA} returns `Yes'. The reachability tree of the rule vector ${\mathcal{R}}$ is shown in Fig.~\ref{NPCAweight}. In the reachability tree, we can also see that in the leaves, all the RMTs have zero-weight.
\end{example}

\section{Synthesis of NCCA}
\label{Synthesis_ncca}

{\em Synthesis} is the converse of decision problem. Here we need to find individual rules to get an $n$-cell NCCA $\langle {\mathcal{R}}_0, {\mathcal{R}}_1, \cdots, {\mathcal{R}}_{n-1} \rangle$. The problem can be stated as following :
\begin{itemize}
\item[]{\em Given a finite $n \geq 5$, select the individual cell rules to get a rule vector of an $n$-cell NCCA.}
\end{itemize}
To solve this problem, we utilize the theories developed in previous sections. As a first step, we arbitrarily choose a number conserving rule as ${\mathcal{R}}_0$. Then we choose another number conserving rule which can act as ${\mathcal{R}}_1$. However, to get ${\mathcal{R}}_1$ (and other rules), we set the next state values of individual RMTs of the rule following some conditions, which finally make the CA as NCCA. For example, we set the next states of RMTs 0 and 7 of a rule ${\mathcal{R}}_i$ to be selected, as 0 and 1 respectively due to Lemma~\ref{universal}. Similarly, if weight of RMT $2 \in \Gamma_0^{N_{i.j}}$ is -1, we set ${\mathcal{R}}_i[2]$ as 0 to keep weight of the RMT consistent with Tab.~\ref{possible_weight}.

The synthesis procedure, like the decision algorithm, develops super node ${\mathcal{N}}_{i}$ for each level $i$, finds weights of RMTs in $\Gamma_k^{i}$ ($k \in \{0, 1, 2, 3\}$), and based on the weights, it {\em synthesizes} ${\mathcal{R}}_i$. However, we need to take special care for choosing ${\mathcal{R}}_1$, ${\mathcal{R}}_{n-2}$ and ${\mathcal{R}}_{n-1}$ due to some reasons which will be discussed later. Next we discuss how we can synthesize ${\mathcal{R}}_i$ by setting the next state values of RMTs. For ease of understanding, we provide four flowcharts (Fig.~\ref{flowchat2}, Fig.~\ref{flowchat1}, Fig.~\ref{flowchat4} and Fig.~\ref{flowchat3}), which summarize the selection techniques of ${\mathcal{R}}_i$, ${\mathcal{R}}_1$, ${\mathcal{R}}_{n-1}$ and ${\mathcal{R}}_{n-2}$.


\subsection{Selecting ${\mathcal{R}}_i$}

As stated before, we set the next states of RMTs to get ${\mathcal{R}}_i$. As a first step, we set the following: ${\mathcal{R}}_i[0] \leftarrow 0$ and ${\mathcal{R}}_i[7] \leftarrow 1$ (Lemma~\ref{universal}). Now according to Corollary~\ref{coro_40}, for any $k$, $W_k(i, 4)$ $\geq$ $W_k(i, 0)$. However, if $W_k(i, 4)$ = $W_k(i, 0)$ then in an NCCA, ${\mathcal{R}}_i[4] = {\mathcal{R}}_i[0]$ (Corollary~\ref{equivalentTh1}). That is, to get NCCA, we need to set ${\mathcal{R}}_i[4] \leftarrow 0$. On the other hand, if $W_k(i, 4)$ $>$ $W_k(i, 0)$, then set ${\mathcal{R}}_i[4] \leftarrow 1$ (Corollary~\ref{1110}). Interestingly, if $W_k(i, 4)$ $>$ $W_k(i, 0)$ for one value of $k$, say $k$ = 0, then $W_k(i, 4)$ $>$ $W_k(i, 0)$ for all other $k$, that is for $k$ = 1, 2 and 3. 

As RMTs 1 and 5 are sibling of RMTs 0 and 4 respectively, so if ${\mathcal{R}}_i[4] \neq {\mathcal{R}}_i[0]$ then we set ${\mathcal{R}}_i[5] \leftarrow 1$ and ${\mathcal{R}}_i[1] \leftarrow 0$ using Theorem~\ref{rule_select}. However, if ${\mathcal{R}}_i[4] = {\mathcal{R}}_i[0]$, we can not set next state values of RMTs 1 and 5 following Theorem~\ref{rule_select}. Hence when ${\mathcal{R}}_i[4] = {\mathcal{R}}_i[0]$, to set the next state values of RMTs 1 and 5, we need to check some additional conditions relating weights of these RMTs. Please recall the following conditions of Section~\ref{weight_RMT}, which are to be followed by any NCCAs.

\begin{center}
(i) If $W_0(i, 5)$ = 1 then ${\mathcal R}_{i}[5]$ = 1
~~~~~~~~~~~~~~~~(ii) If $W_1(i, 5)$ = 2 then ${\mathcal R}_{i}[5]$ = 1\\
(iii) If $W_2(i, 5)$ = 1 then ${\mathcal R}_{i}[5]$ = 1
~~~~~~~~~~~~~~~(iv) If $W_3(i, 5)$ = 1 then ${\mathcal R}_{i}[5]$ = 1\\
(v) If $W_3(i, 1)$ = 1 then ${\mathcal R}_{i}[1]$ =1
~~~~~~~~~~~~~~~~~(vi) If $W_3(i, 1)$ = -1 then ${\mathcal R}_{i}[1]$ = 0 
\end{center}

Therefore, when any of the conditions (i) to (iv) is true, we set ${\mathcal R}_{i}[5]$ $\leftarrow$ 1. Following Theorem~\ref{rule_select}, we can set ${\mathcal R}_{i}[1]$ $\leftarrow$ 1. If condition (v) is true,  we set ${\mathcal R}_{i}[1]$ $\leftarrow$ 1 and to respect Theorem~\ref{rule_select}, we set ${\mathcal R}_{i}[5]$ $\leftarrow$ 1. That is, if any of the conditions (i) to (v) is true, we set ${\mathcal R}_{i}[1]$ $\leftarrow$ 1, ${\mathcal R}_{i}[5]$ $\leftarrow$ 1. Similarly, when condition (vi) is true, then we set as ${\mathcal R}_{i}[1]$ $\leftarrow$ 0 and ${\mathcal R}_{i}[5]$ $\leftarrow$ 0. However, one may notice that when condition (vi) is true, none of the conditions (i) to (v) can be true. Otherwise, weight of RMTs will become inconsistent.


Following the same rationale, we can set the next state values of rest RMTs of ${\mathcal{R}}_i$. According to Corollary~\ref{coro_37}, for any $k$, $W_k(i, 3)$ $\leq$ $W_k(i, 7)$. So if $W_k(i, 3)$ = $W_k(i, 7)$, then we set ${\mathcal{R}}_i[3] \leftarrow 1$ (using Corollary~\ref{equivalentTh1}), otherwise set ${\mathcal{R}}_i[3] \leftarrow 0$ (using Corollary~\ref{1110}). As RMTs 2 and 6 are the sibling of RMTs 3 and 7 respectively, so if ${\mathcal{R}}_i[3] \neq {\mathcal{R}}_i[7]$ then we set ${\mathcal{R}}_i[2] \leftarrow 0$ and ${\mathcal{R}}_i[6] \leftarrow 1$ (Theorem~\ref{rule_select}). When ${\mathcal{R}}_i[3] = {\mathcal{R}}_i[7]$ we can not set next state values of RMTs 2 and 6 using Theorem~\ref{rule_select}, as before. However, following conditions are already noted in Section~\ref{weight_RMT}.

\begin{center}
(a) If $W_0(i, 2)$ = -1 then ${\mathcal R}_{i}[2]$ = 0
~~~~~~~~~~~~~~~~(b) If $W_0(i, 6)$ = -1 then ${\mathcal R}_{i}[6]$ = 0\\ 
(c) If $W_1(i, 2)$ = -1 then ${\mathcal R}_{i}[2]$ = 0
~~~~~~~~~~~~~~~(d) If $W_2(i, 2)$ = -2 then ${\mathcal R}_{i}[2]$ = 0\\
(e) If $W_3(i, 2)$ = -1 then ${\mathcal R}_{i}[2]$ = 0
~~~~~~~~~~~~~~~~(f) If $W_0(i, 6)$ = 1 then ${\mathcal R}_{i}[6]$ = 1
\end{center}

If any of the conditions (a) to (e) is true, then our synthesis algorithms sets ${\mathcal R}_{i}[2]$ $\leftarrow$ 0 and ${\mathcal R}_{i}[6]$ $\leftarrow$ 0. Similarly, if condition (f) is true then we set ${\mathcal R}_{i}[2]$ $\leftarrow$ 1 and ${\mathcal R}_{i}[6]$ $\leftarrow$ 1. In Algorithm~\ref{Syn_NCCA}, the proposed synthesis algorithm, Step 6 and Step 7 use these ideas to get ${\mathcal{R}}_i$.
 
After taking all the above measures, it may so happens that some of the RMTs of ${\mathcal{R}}_i$ remain unfilled. In that case, we set unfilled RMTs randomly obeying the condition of Theorem~\ref{rule_select}. However, we summarize the selection procedure of ${\mathcal{R}}_i$ in the flowchart of Fig.~\ref{flowchat2}.
 
\begin{figure}
\begin{center}
\includegraphics[height=5.8in, width=5.9in]{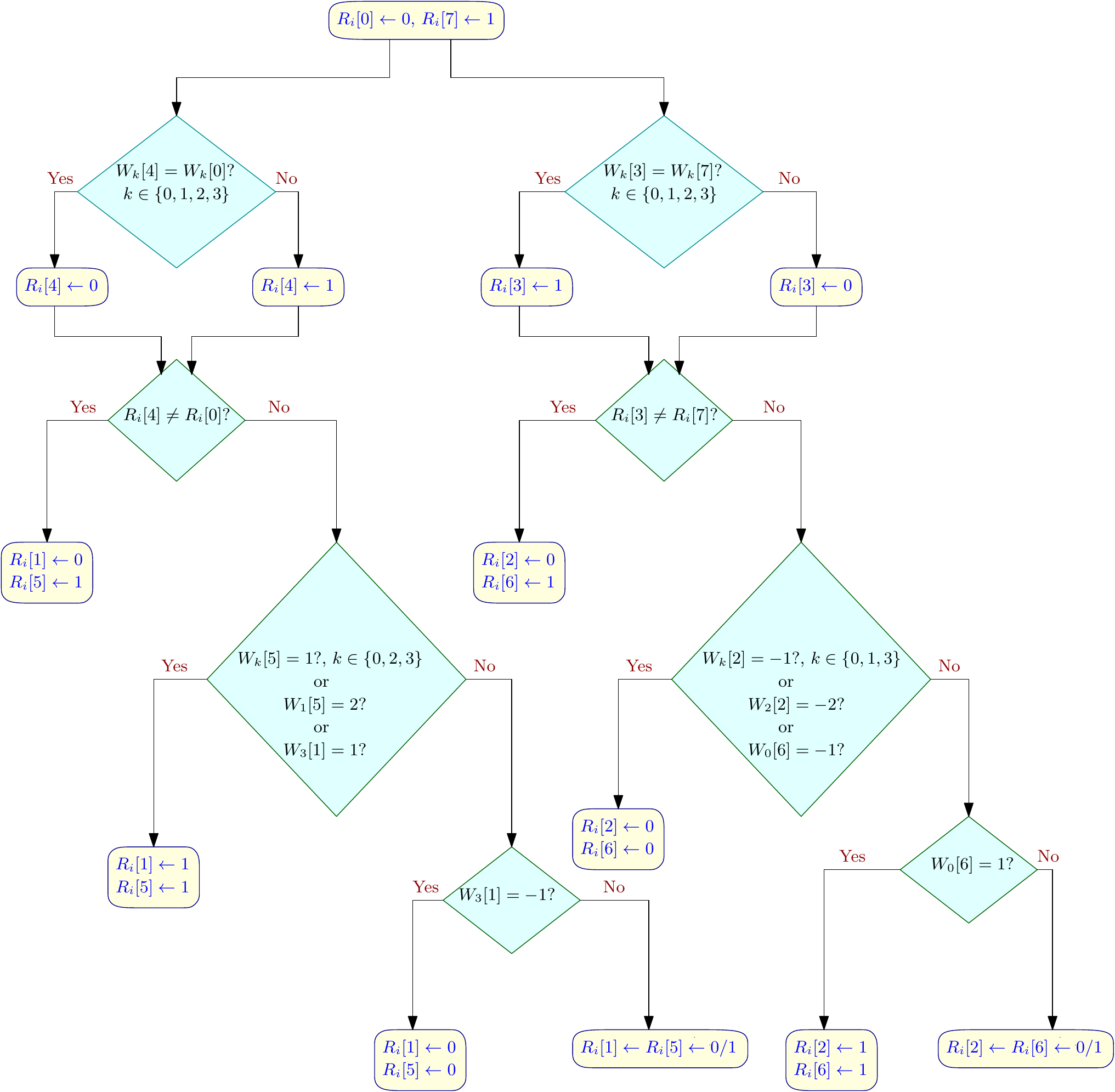}
\caption{Flowchart for the selection of ${\mathcal{R}}_i$}
\label{flowchat2}
\end{center}
\end{figure}

\subsection{Selecting ${\mathcal{R}}_1$}

Though it is not mentioned explicitly, but selection process of a rule ${\mathcal{R}}_i$ expects the presence of all the RMTs in a set $\Gamma_k^{i}$ of a super node ${\mathcal{N}}_{i}$. For example, if we want to check whether $W_k(i, 4)$ is greater than $W_k(i, 0)$, then RMTs 0 and 4 should be in $\Gamma_k^{{i}}$. However, in the super node ${\mathcal{N}}_{1}$, RMT 0 is in $\Gamma_0^1$ but RMT 4 is not. In fact, all of the 8 RMTs are not present in any $\Gamma_k^{1}$ (see Fig.~\ref{NPCAweight} for verification). Therefore, if we apply the selection process of ${\mathcal{R}}_i$ to select ${\mathcal{R}}_1$, next state value of many RMTs may remain unfilled. If the unfilled positions are filled up arbitrary, some inconsistencies in weights may arise. Hence, we take some extra measures while choosing the rule ${\mathcal{R}}_1$.

What we do in this case is, we relate weights of RMTs, present in different sets of ${\mathcal{N}}_{1}$. As always, we set ${\mathcal R}_{1}[0]$ $\leftarrow$ 0 and ${\mathcal R}_{1}[7]$ $\leftarrow$ 0. However, following four conditions, out of sixteen, reported in section~\ref{weight_RMT}, can arise in case of ${\mathcal{N}}_{1}$, when ${\mathcal R}_{0}$ is a number conserving rule.
\begin{center}
(i) If $W_0(1, 2) = -1$ then ${\mathcal R}_{1}[2] = 0$~~~~~~~~~(ii) If $W_1(1, 4) = 0$ then ${\mathcal R}_{1}[4] = 0$\\
(iii) If $W_2(1, 3) = 0$ then ${\mathcal R}_{1}[3] = 1$~~~~~~~~~(iv) If $W_3(1, 5) = 1$ then ${\mathcal R}_{1}[5] = 1$
\end{center}

If any of the above conditions arise, we can set the next state value of the corresponding RMT directly. For example, if $W_1(1, 4) = 0$, to get an NCCA, we need to set ${\mathcal R}_{1}[4]$ $\leftarrow$ 0. Interestingly, conditions (i) and (iii)(similarly, condition (ii) and (iv)) are mutually exclusive -- that is, both can not be true simultaneously. However, in the proposed synthesis algorithm, we start by setting the next state value of RMT 2. Here, Corollary~\ref{basic3} and Theorem~\ref{rule_select} are instrumental. However, in this discussion, we mainly use the RMTs and their weights of two sets -- $\Gamma_0^1$ and $\Gamma_1^1$. One can reach to same result, if she uses other sets, like $\Gamma_2^1$ and $\Gamma_3^1$.

As before, if $W_0(1, 2)$ = -1, we set ${\mathcal R}_{1}[2]$ $\leftarrow$ 0. Hence, $W_0(2, 4) = W_0(2, 5)$ = 0. That is, $W_0(2, 0) - W_0(2, 4)$ = 0. According to Corollary~\ref{basic3}, $W_1(2, 0) - W_1(2, 4)$ is also to be 0. To get this, we set ${\mathcal R}_{1}[4]$ $\leftarrow$ 0 and ${\mathcal R}_{1}[6]$ $\leftarrow$ 1 when $W_1(1, 4) = W_1(1, 6)$. When $W_1(1, 4) = W_1(1, 6) = 0$, then condition (ii) of the above is true. However, we need not to check this condition, as we have already set ${\mathcal R}_{1}[4]$ $\leftarrow$ 0 using Corollary~\ref{basic3}. One may notice that if $W_1(1, 4) = W_1(1, 6)$ then $W_3(1, 4) = W_3(1, 6)$ when ${\mathcal{R}}_0$ is a number conserving rule. When ${\mathcal{R}}_1[6]$ is set as 1, then using Theorem~\ref{rule_select}, we set ${\mathcal R}_{1}[3]$ $\leftarrow$ 0. The rest RMTs, that is RMTs 1 and 5, are arbitrarily set either as ${\mathcal R}_{1}[1]$ $\leftarrow$ ${\mathcal R}_{1}[5]$ $\leftarrow$ 0 or as ${\mathcal R}_{1}[1]$ $\leftarrow$ ${\mathcal R}_{1}[5]$ $\leftarrow$ 1. 

But when $W_1(1, 4) \neq W_1(1, 6)$, then the only possibility is: $W_1(1, 4)$ = 1 and  $W_1(1, 6)$ = 0 (and, $W_3(1, 4)$ = 1 and $W_3(1, 6)$ = 0). Further, $W_3(1, 5) = W_3(1, 4) = 1$. So, condition (iv) of the above is arise, and we set ${\mathcal R}_{1}[5]$ $\leftarrow$ 1. To meet the condition of Corollary~\ref{basic3} at level 2, we need to arbitrarily set either as ${\mathcal R}_{1}[4]$ $\leftarrow$ ${\mathcal R}_{1}[6]$ $\leftarrow$ 0 or as ${\mathcal R}_{1}[4]$ $\leftarrow$ ${\mathcal R}_{1}[6]$ $\leftarrow$ 1, because when condition (iv) is true then condition (ii) can not be true. As a next step, using Theorem~\ref{rule_select}, we set ${\mathcal R}_{1}[1]$ $\leftarrow 1$ and ${\mathcal R}_{1}[3]$ $\leftarrow 1$ if ${\mathcal R}_{1}[4]$ = 0, set ${\mathcal R}_{1}[1]$ $\leftarrow 0$ and ${\mathcal R}_{1}[3]$ $\leftarrow$ 0 otherwise.

Now if $W_0(1, 2)$ $\neq -1$, that is, if $W_0(1, 2)$ $= 0$, ${\mathcal R}_{1}[2]$ can be anything 0 or 1. Let us set ${\mathcal R}_{1}[2]$ $\leftarrow$ 0. Then, $W_0(2, 4) = W_0(2, 5) = 1$, and $W_0(2, 0) - W_0(2, 4) = -1$. Like the previous case, we need to set same next state value to ${\mathcal R}_{1}[4]$ and ${\mathcal R}_{1}[6]$ when $W_1(1, 4) = W_1(1, 6)$, in order to get $W_1(2, 0) - W_1(2, 4) = -1$ (otherwise Corollary~\ref{basic3} will be violated). But if $W_1(1, 4) = 0$ then to satisfy condition (ii) of the above, we need to set ${\mathcal R}_{1}[4]$ $\leftarrow$ ${\mathcal R}_{1}[6]$ $\leftarrow$ 0. Now using Theorem~\ref{rule_select}, we set ${\mathcal R}_{1}[3]$ $\leftarrow$ 1. On the other hand, if $W_2(1, 3) = 0$, then following condition (iii) of the above, we set ${\mathcal R}_{1}[3]$ $\leftarrow$ 1. Using Theorem~\ref{rule_select} again, we need to set ${\mathcal R}_{1}[6]$ $\leftarrow$ 0, and hence ${\mathcal R}_{1}[4]$ $\leftarrow$ ${\mathcal R}_{1}[6]$ $\leftarrow$ 0. However, if $W_1(1, 4) = W_1(1, 6) = 1$, then arbitrarily set either as ${\mathcal R}_{1}[4]$ $\leftarrow$ ${\mathcal R}_{1}[6]$ $\leftarrow$ 0 or ${\mathcal R}_{1}[4]$ $\leftarrow$ ${\mathcal R}_{1}[6]$ $\leftarrow$ 1, and RMT 3 is set following Theorem~\ref{rule_select}. However, if $W_1(1, 4) = 1$ and $W_1(1, 6) = 0$ (that is, when $W_1(1, 4) \neq W_1(1, 6)$), then to meet the condition of Corollary~\ref{basic3} at level 2, we need to set ${\mathcal R}_{1}[4]$ $\leftarrow$ 1 and ${\mathcal R}_{1}[6]$ $\leftarrow$ 0. When $W_1(1, 4) \neq W_1(1, 6)$, then $W_3(1, 5) = 1$, so we set ${\mathcal R}_{1}[5]$ $\leftarrow$ 1. Now using Theorem~\ref{rule_select}, we can set ${\mathcal R}_{1}[1]$ $\leftarrow$ 0 and ${\mathcal R}_{1}[3]$ $\leftarrow$ 1.

Now let us set ${\mathcal R}_{1}[2]$ $\leftarrow 1$ when $W_0(1, 2) = 0$. Using Theorem~\ref{rule_select}, we can set ${\mathcal R}_{1}[3]$ $\leftarrow 1$ and ${\mathcal R}_{1}[6]$ $\leftarrow 1$. When $W_0(1, 2) = 0$, then, $W_0(2, 4) = W_0(2, 5) = 0$ and $W_0(2, 0) - W_0(2, 4) = 0$. A similar case has already been discussed before. Following that rationale, we set the rest RMTs as following: If $W_1(1, 4) = W_1(1, 6)$ then set ${\mathcal R}_{1}[4]$ $\leftarrow 0$, otherwise set ${\mathcal R}_{1}[4]$ $\leftarrow 1$. Now if ${\mathcal R}_{1}[4] = 0$ then arbitrarily set either as ${\mathcal R}_{1}[1]$ $\leftarrow$ ${\mathcal R}_{1}[5]$ $\leftarrow 0$ or as ${\mathcal R}_{1}[1]$ $\leftarrow$ ${\mathcal R}_{1}[5]$ $\leftarrow 1$, otherwise set ${\mathcal R}_{1}[1]$ $\leftarrow 0$ and ${\mathcal R}_{1}[5]$ $\leftarrow 1$.

All the points, noted above, are summarized in flowchart of Fig.~\ref{flowchat1}.

\begin{figure}
\begin{center}
\includegraphics[height=5.8in, width=5.9in]{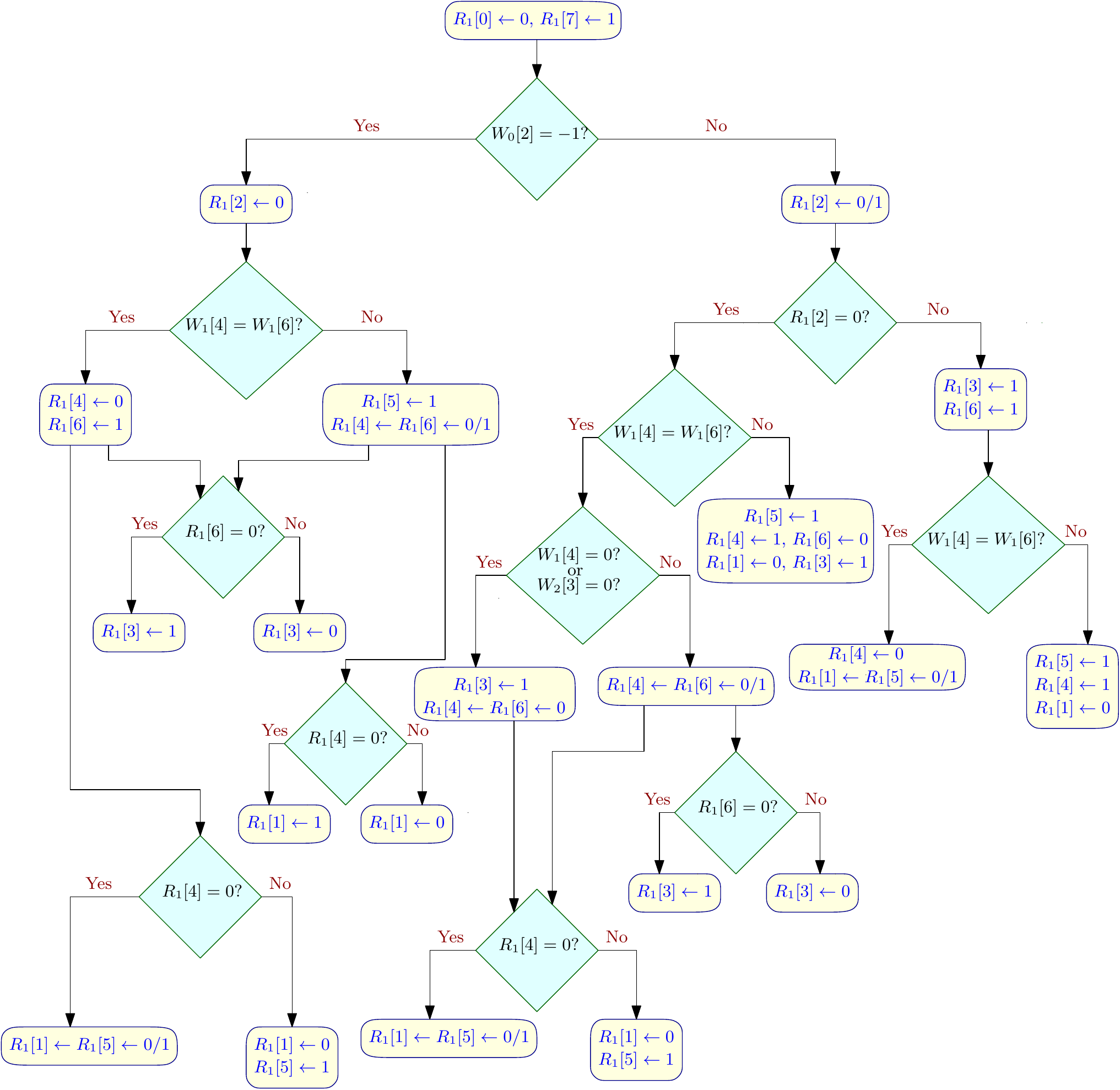}
\caption{Flowchart for the selection of ${\mathcal{R}}_1$}
\label{flowchat1}
\end{center}
\end{figure} 

\subsection{Selecting ${\mathcal{R}}_{n-2}$ and ${\mathcal{R}}_{n-1}$}

In case of selecting of ${\mathcal{R}}_{n-2}$ and ${\mathcal{R}}_{n-1}$, we can not use the general technique of choosing a rule ${\mathcal{R}}_i$. Because, if we do so, some of the RMTs at leaves of the reachability tree may attain non-zero weights. So we impose some additional restrictions while ${\mathcal{R}}_{n-2}$ and ${\mathcal{R}}_{n-1}$ are chosen.

\vspace{2mm}
\noindent
\textbf{Selection of ${\mathcal{R}}_{n-1}$:} All the RMTs can not be present in $\Gamma_k^{N_{n-1.j}}$ of an arbitrary node of level $n-1$ for any $k \in \{0, 1, 2, 3\}$ (see point~\ref{rtd6} of Definition.~\ref{Rtree_def}), and in general, the super node ${\mathcal{N}}_{n-1}$ = ($\{0, 4\}, \{1, 5\}, \{2, 6\}, \{3, 7\}$). As a first step, we set ${\mathcal R}_{n-1}[0]$ $\leftarrow 0$ and ${\mathcal R}_{n-1}[7]$ $\leftarrow 1$ (Lemma~\ref{universal}). However, the RMTs of $\Gamma_k^{n-1}$ can not attain all the possible weights which are noted in Tab.~\ref{possible_weight}. For example, no RMT of $\Gamma_k^{n-1}$ can have weight 2 or -2. Possible weights of RMTs at ${\mathcal{N}}_{n-1}$ are following:

\hspace{18mm} (i) $W_0(n-1, 4)$ = 0 or 1, ~~~~~~~~~~~(ii) $W_1(n-1, 1)$ = 0 or 1,

\hspace{18mm}(iii) $W_1(n-1, 5)$ = 0 or 1, ~~~~~~~~~~(iv) $W_2(n-1, 2)$ = 0 or -1,

\hspace{18mm}(v) $W_2(n-1, 6)$ = 0 or -1  ~~~~~~~~~~~(vi) $W_3(n-1, 3)$ = 0 or -1. 

This indicates that we have to select ${\mathcal{R}}_{n-2}$ in such a way that the RMTs can attain above possible weights.

As usual, if $W_0(n-1, 4) = 0$ then set ${\mathcal R}_{n-1}[4]$ $\leftarrow 0$, otherwise set ${\mathcal R}_{n-1}[4]$ $\leftarrow 1$. Similarly, set ${\mathcal R}_{n-1}[3]$ $\leftarrow 1$ if $W_3(n-1, 3) = 0$, otherwise set ${\mathcal R}_{n-1}[3]$ $\leftarrow 0$. However, we impose following obvious restrictions during selection of ${\mathcal{R}}_{n-1}$. 
\begin{enumerate}
\item If $W_1(n-1,r)$ = 0 where $r \in \{1, 5\}$ then set ${\mathcal R}_{n-1}[r]$ $\leftarrow$ 0, otherwise set ${\mathcal R}_{n-1}[r]$ $\leftarrow$ 1. 
\item If $W_2(n-1,r)$ = 0 where $r \in \{2, 6\}$ then set ${\mathcal R}_{n-1}[r]$ $\leftarrow$ 1, otherwise set  ${\mathcal R}_{n-1}[r]$ $\leftarrow$ 0. 
\end{enumerate} 

The flowchart of Fig.~\ref{flowchat4} summarizes the selection procedure of ${\mathcal{R}}_{n-1}$.

\begin{figure}[h]
\begin{center}
\includegraphics[height=3.1in, width=5.9in]{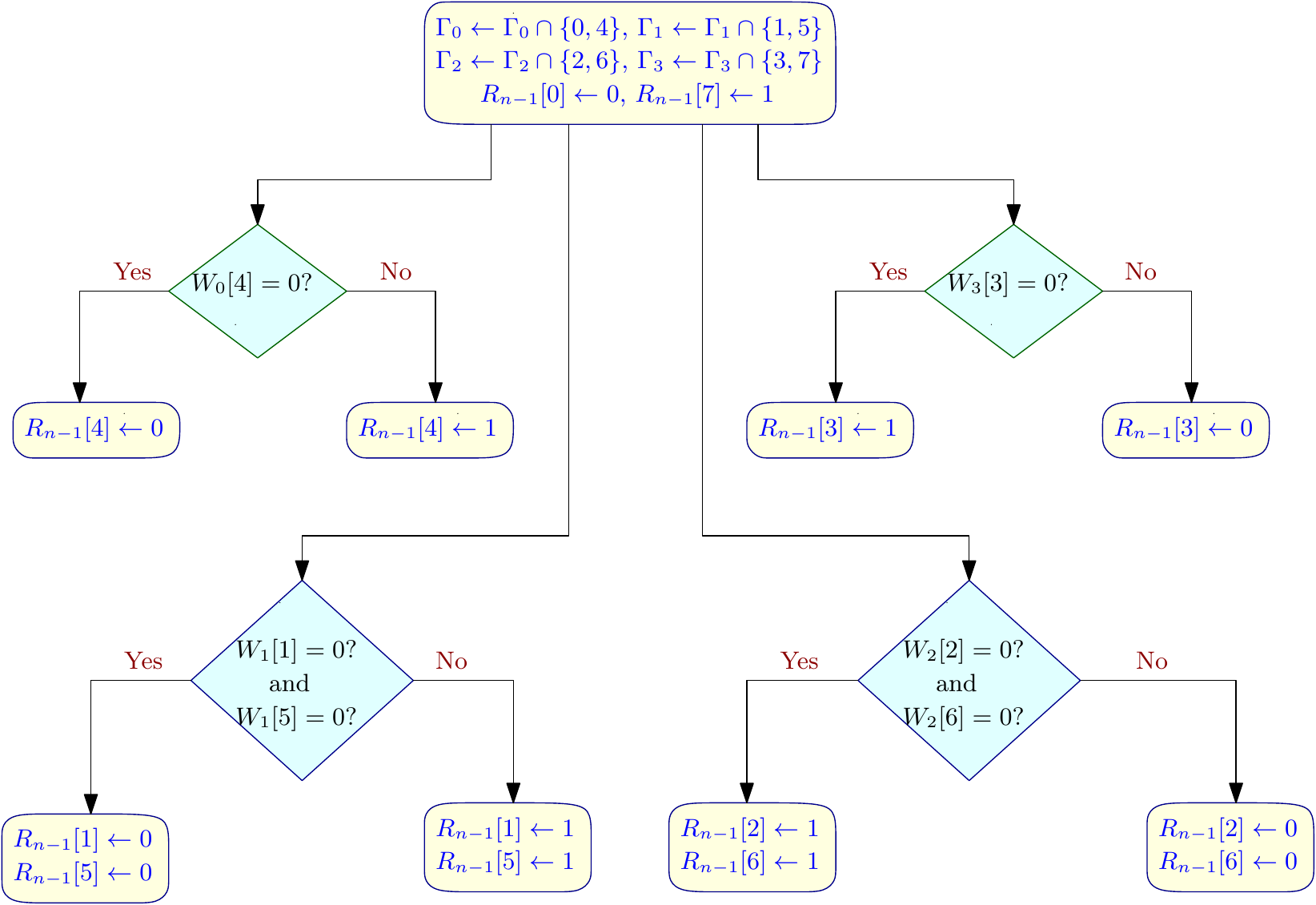}
\caption{Flowchart for the selection of ${\mathcal{R}}_{n-1}$}
\label{flowchat4}
\end{center}
\end{figure} 

\vspace{2mm}
\noindent
\textbf{Selection of ${\mathcal{R}}_{n-2}$:} According to point.~\ref{rtd5} of Definition~\ref{Rtree_def}, all the RMTs are not present in $\Gamma_k^{{n-2.j}}$ of an arbitrary node of level $n-2$ for any $k \in \{0, 1, 2, 3\}$. In general, the super node ${\mathcal{N}}_{n-1}$ = ($\{0, 2, 4, 6\}, \{0, 2, 4, 6\}$, $\{1, 3, 5, 7\}, \{1, 3, 5, 7\}$). Though the RMTs of ${\mathcal{N}}_{n-2}$ can attain all possible weights as described in Tab.~\ref{possible_weight}, we have to select the next state values of RMTs of ${\mathcal{R}}_{n-2}$ in such a fashion that the possible weights of RMTs of ${\mathcal{N}}_{n-1}$ can be the above stated weights only. However, if following conditions arise, we take the same actions which we have taken in selecting ${\mathcal{R}}_{i}$:

\begin{enumerate}
\item \label{n-2_1}If $W_k(n-2,4)$ = $W_k(n-2,0)$, $k \in \{0, 1\}$ then set ${\mathcal R}_{n-2}[4] \leftarrow 0$. Otherwise set ${\mathcal R}_{n-2}[4] \leftarrow 1$.

\item \label{n-2_2}If $W_k(n-2,3)$ = $W_k(n-2,7)$, $k \in \{2, 3\}$ then set ${\mathcal R}_{n-2}[3] \leftarrow 1$. Otherwise set ${\mathcal R}_{n-2}[3] \leftarrow 0$. 

\item \label{n-2_3}If $W_k(n-2,2)$ $<$ $W_k(n-2,6)$, $k \in \{0, 1\}$ then set ${\mathcal R}_{n-2}[2] \leftarrow 0$ and ${\mathcal R}_{n-2}[6] \leftarrow 1$. 

\item \label{n-2_4}If $W_k(n-2,5)$ $>$ $W_k(n-2,1)$, $k \in \{2, 3\}$ then set ${\mathcal R}_{n-2}[5] \leftarrow 1$ and ${\mathcal R}_{n-2}[1] \leftarrow 0$.

\end{enumerate}

If all the conditions, related to weight, arise, we can set the next states of all the RMTs of ${\mathcal{R}}_{n-2}$. However, if one or more conditions are not met, some RMTs next state values remain unfilled, and then we need to take special care due to the reasons stated above. Following are the conditions, which if arise, we specially set the next state values of RMTs. Obviously, these conditions are not specially dealt in selecting ${\mathcal{R}}_{i}$.

\begin{enumerate}
\item \label{n-2_5} If $W_0(n-2,2)$ = $W_0(n-2,6)$ = -1 then set ${\mathcal R}_{n-2}[2]$ $\leftarrow$ 0 and ${\mathcal R}_{n-2}[6]$ $\leftarrow$ 0.

Because both the RMTs contribute RMT 4 at level $(n-1)$ with weight $W_0(n-1,4)$ = 0. If ${\mathcal R}_{n-2}[2]$ and ${\mathcal R}_{n-2}[6]$ are 1, then they contribute RMT 4 at level $(n-1)$ with weight $W_0(n-1,4)$ = -1, which can not contribute RMT 0, 1 at leaf nodes with weight 0. 

\item \label{n-2_6}If $W_1(n-2,2)$ = $W_1(n-2,6)$ = 1 then set ${\mathcal R}_{n-2}[2]$ $\leftarrow$ 1 and ${\mathcal R}_{n-2}[6]$ $\leftarrow$ 1. 

Because both the RMTs contribute RMT 5 at level $(n-1)$ with weight $W_1(n-1,5)$ = 1. If ${\mathcal R}_{n-2}[2]$ and ${\mathcal R}_{n-2}[6]$ are 0, then they contribute RMT 5 at level $(n-1)$ with weight $W_1(n-1,5)$ = 2, which can not contribute RMT 2, 3 at leaf nodes with weight 0. 

\item \label{n-2_7}If $W_2(n-2,1)$ = $W_2(n-2,5)$ = -1 then set ${\mathcal R}_{n-2}[1]$ $\leftarrow$ 0 and ${\mathcal R}_{n-2}[5]$ $\leftarrow$ 0.

Because both the RMTs contribute RMT 2 at level $(n-1)$ with weight $W_2(n-1,2)$ = -1. If ${\mathcal R}_{n-2}[1]$ and ${\mathcal R}_{n-2}[5]$ are 1, then they contribute RMT 2 at level $(n-1)$ with weight $W_2(n-1,2)$ = -2, which can not contribute RMT 4, 5 at leaf nodes with weight 0. 

\item \label{n-2_8}If $W_3(n-2,1)$ = $W_3(n-2,5)$ = 1 then set ${\mathcal R}_{n-2}[1]$ $\leftarrow$ 1 and ${\mathcal R}_{n-2}[5]$ $\leftarrow$ 1. 

Because they contribute RMT 3 at level $(n-1)$ with weight $W_3(n-1,3)$ = 0. If ${\mathcal R}_{n-2}[1]$ and ${\mathcal R}_{n-2}[5]$ are 0, then they contribute RMT 3 at level $(n-1)$ with weight $W_3(n-1,3)$ = 1, which can not contribute RMT 6, 7 at leaf nodes with weight 0.

\item \label{n-2_9}If $W_k(n-2,2)$ = $W_k(n-2,6)$ = 0 where $k \in \{0, 1\}$, then arbitrarily set either as ${\mathcal R}_{n-2}[2]$ $\leftarrow$ ${\mathcal R}_{n-2}[6]$ $\leftarrow$ 0 or as ${\mathcal R}_{n-2}[2]$ $\leftarrow$ ${\mathcal R}_{n-2}[6]$ $\leftarrow$ 1. 

\item \label{n-2_10}If $W_k(n-2,1)$ = $W_k(n-2,5)$ = 0 where $k \in \{2, 3\}$, then arbitrarily set either as ${\mathcal R}_{n-2}[1]$ $\leftarrow$ ${\mathcal R}_{n-2}[5]$ $\leftarrow$ 0 or as ${\mathcal R}_{n-2}[1]$ $\leftarrow$ ${\mathcal R}_{n-2}[5]$ $\leftarrow$ 1. 
\end{enumerate}

One can verify that conditions~\ref{n-2_5} and \ref{n-2_6} can not be true together. Similarly, conditions~\ref{n-2_7} and \ref{n-2_8} are not true together. Hence, we are able to set next state values of all the RMTs of ${\mathcal{R}}_{n-2}$. For easy reference, please see the flowchart of Fig.~\ref{flowchat3} which summarizes the selection of ${\mathcal{R}}_{n-2}$.

\begin{figure}
\begin{center}
\includegraphics[height=5.8in, width=5.9in]{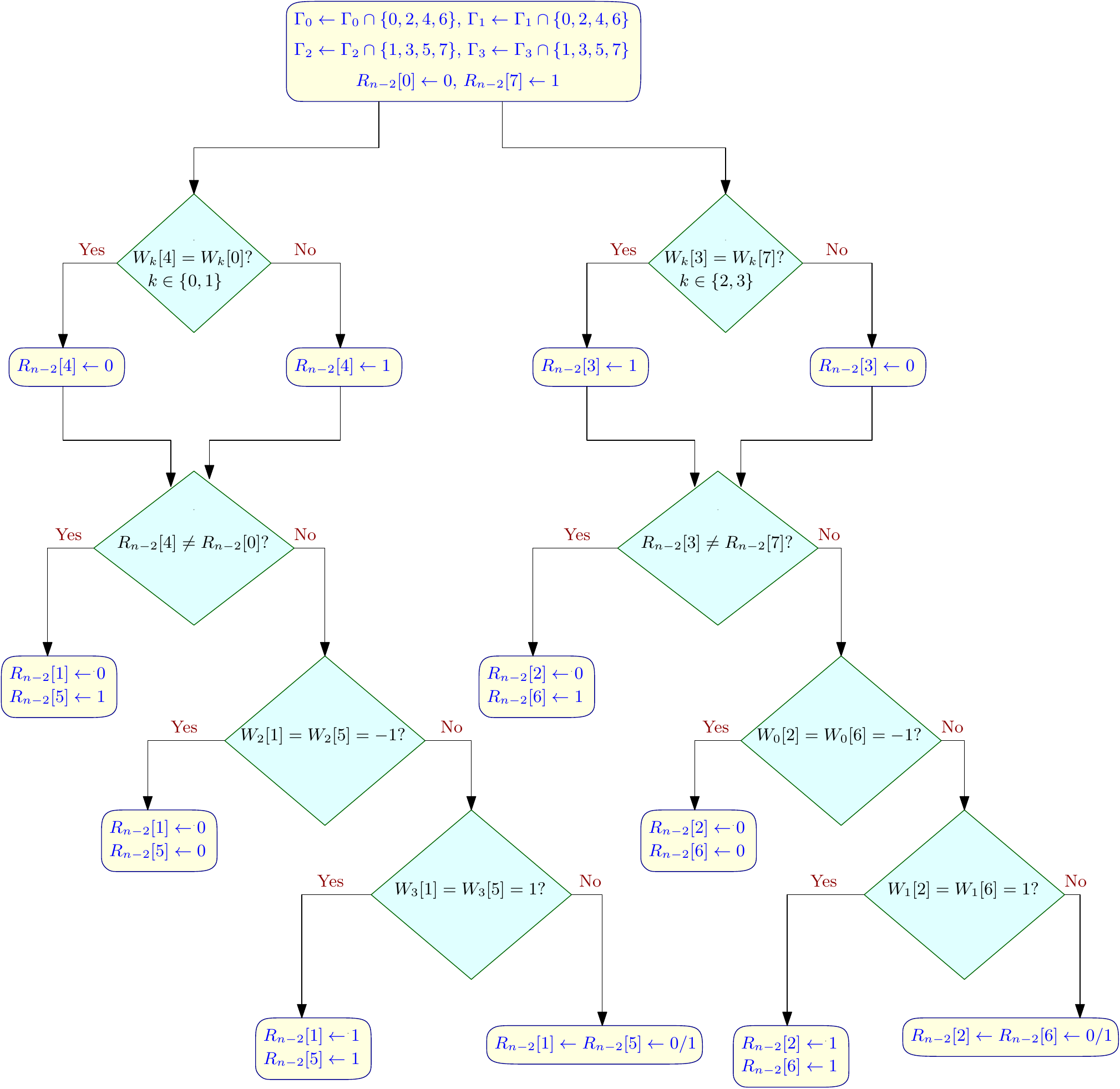}
\caption{Flowchart for the selection of ${\mathcal{R}}_{n-2}$}
\label{flowchat3}
\end{center}
\end{figure}

\subsection{Algorithm}

\begin{algorithm}[hbtp]
\scriptsize
\SetKwInOut{Input}{Input}
\SetKwInOut{Output}{Output}

\Input{$n (\geq 5)$}
\Output{NCCA with rule vector ${\mathcal{R}}= \langle{\mathcal{R}}_0, {\mathcal{R}}_1, \cdots, {\mathcal{R}}_{n-1}\rangle$}

\rule[5pt]{0.95\textwidth}{0.85pt}\\

\Begin{
	\textbf{Step 1 :}\label{st1} Set $i \leftarrow 2$,  $\Gamma_0 \leftarrow \{0, 1\}$, $\Gamma_1 \leftarrow \{2, 3\}$, $\Gamma_2 \leftarrow \{4, 5\}$, $\Gamma_3 \leftarrow \{6, 7\}$, $W_k[r] \leftarrow 0$ if $r \in \Gamma_k$, for each $k$, $0 \leq k \leq 3$\;

	\textbf{Step 2 :}\label{st2} Set ${\mathcal{R}}_j[0] \leftarrow 0$, ${\mathcal{R}}_j[7] \leftarrow 1$, $0 \leq j \leq n-1$\;
	
	\textbf{Step 3 :}\label{st3} Randomly select a rule ${\mathcal{R}}_0$ from rule set \{136, 170, 184, 192, 204, 226, 238, 240, 252\}\;	
\hspace{9mm}Call {\it FindNextWeight(0, $\Gamma_k$, $W_k$)} for each $k \in \{0, 1, 2, 3\}$\;

	 \textbf{Step 4 :}\label{st4} Select the next state values of RMTs of ${\mathcal{R}}_1$ according to the following:\\
				\eIf{$W_0[2] = -1$}
					{
						${\mathcal{R}}_1[2]$ $\leftarrow 0$\;
						\eIf{$W_1[4] = W_1[6]$}
						{
							${\mathcal{R}}_1[3] \leftarrow {\mathcal{R}}_1[4] \leftarrow 0$, ${\mathcal{R}}_1[6]$ $\leftarrow 1$ \;
							${\mathcal{R}}_1[1] \leftarrow {\mathcal{R}}_1[5] \leftarrow \alpha$ where $\alpha \in \{0, 1\}$ \;
						}
						{
							${\mathcal{R}}_1[5]$ $\leftarrow 1$, ${\mathcal{R}}_1[4] \leftarrow {\mathcal{R}}_1[6] \leftarrow \alpha$ where $\alpha \in \{0, 1\}$\;
							\eIf{${\mathcal{R}}_1[4] = 0$}
							{
								${\mathcal{R}}_1[1] \leftarrow {\mathcal{R}}_1[3] \leftarrow 1$\;
							}
							{
								${\mathcal{R}}_1[1] \leftarrow {\mathcal{R}}_1[3] \leftarrow 0$\;
							}						
						}
					}
						{
							${\mathcal{R}}_1[2]$ $\leftarrow \alpha$ where $\alpha \in \{0, 1\}$\;
							\eIf{${\mathcal{R}}_1[2] = 0$}
							{
								\eIf{$W_1[4] = W_1[6]$}
								{
									\eIf{$W_1[4] = 0$ or $W_2[3] = 0$}
									{
										${\mathcal{R}}_1[3] \leftarrow 1$, ${\mathcal{R}}_1[4] \leftarrow {\mathcal{R}}_1[6] \leftarrow 0$\;
									}
									{															
											${\mathcal{R}}_1[4] \leftarrow {\mathcal{R}}_1[6] \leftarrow \alpha$ where $\alpha \in \{0, 1\}$ \;
											\eIf{${\mathcal{R}}_1[6] = 0$}
											{
												${\mathcal{R}}_1[3] \leftarrow 1$\;
											}
											{
												${\mathcal{R}}_1[3] \leftarrow 0$\;
											}
									}
									\eIf{${\mathcal{R}}_1[4] = 0$}
									{
										${\mathcal{R}}_1[1] \leftarrow {\mathcal{R}}_1[5] \leftarrow \alpha$ where $\alpha \in \{0, 1\}$ \;
									}
									{
										${\mathcal{R}}_1[1]$ $\leftarrow 0$, ${\mathcal{R}}_1[5]$ $\leftarrow 1$\;
									}
								}
								{
									${\mathcal{R}}_1[5] \leftarrow 1$\;
									${\mathcal{R}}_1[4]$ $\leftarrow 1$, ${\mathcal{R}}_1[6]$ $\leftarrow 0$\;
									${\mathcal{R}}_1[1] \leftarrow 0$, ${\mathcal{R}}_1[3] \leftarrow 1$\;
								}
							}
							{
								${\mathcal{R}}_1[3]$ $\leftarrow 1$, ${\mathcal{R}}_1[6]$ $\leftarrow 1$\;
								\eIf{$W_1[4] = W_1[6]$}
								{
									${\mathcal{R}}_1[4]$ $\leftarrow 0$\;
									${\mathcal{R}}_1[1] \leftarrow {\mathcal{R}}_1[5] \leftarrow \alpha$ where $\alpha \in \{0, 1\}$ \;
								}
								{
									${\mathcal{R}}_1[5] \leftarrow 1$\;
									${\mathcal{R}}_1[4]$ $\leftarrow 1$, ${\mathcal{R}}_1[1] \leftarrow 0$\;
								}	
							}
						}

\textbf{Step 5 :}\label{st5} Call {\it FindNextWeight(1, $\Gamma_k$, $W_k$)} for each $k \in \{0, 1, 2, 3\}$\;
\textbf{Step 6 :}\label{st6} Select the next state values of RMTs 1, 4 and 5 of ${\mathcal{R}}_i$ according to the following:\\
				For any $k \in \{0, 1, 2, 3\}$\\
				\eIf{$W_k[4] = W_k[0]$}
					{
						${\mathcal{R}}_i[4]$ $\leftarrow 0$
					}
					{
						${\mathcal{R}}_i[4]$ $\leftarrow 1$
					}
				\eIf{${\mathcal{R}}_i[4] \neq {\mathcal{R}}_i[0]$}
					{
						${\mathcal{R}}_i[1]$ $\leftarrow 0$, ${\mathcal{R}}_i[5]$ $\leftarrow 1$\;
					}
					{
						\uIf{($W_0[5] = 1$) or ($W_1[5] = 2$) or ($W_2[5] = 1$) or ($W_3[5] = 1$) or ($W_3[1] = 1$)}
					
								{
									${\mathcal{R}}_i[1]$ $\leftarrow 1$, ${\mathcal{R}}_i[5]$ $\leftarrow 1$ \;
								}
						\uElseIf{$W_3[1] = -1$}
								{
									${\mathcal{R}}_i[1]$ $\leftarrow 0$, ${\mathcal{R}}_i[5]$ $\leftarrow 0$\;
								}
						\Else{
									${\mathcal{R}}_i[1] \leftarrow {\mathcal{R}}_i[5] \leftarrow \alpha$ where $\alpha \in \{0, 1\}$ \;
							}
						}		

}

\caption{\emph{Synthesize NCCA}}
\label{Syn_NCCA}
\end{algorithm}

\setcounter{algocf}{0}
\begin{algorithm}[hbtp]
\scriptsize
\SetKwInOut{Input}{Input}
\SetKwInOut{Output}{Output}

		\textbf{Step 7 :}\label{st7} Select the next state values of RMTs 2, 3 and 6 of ${\mathcal{R}}_i$ according to the following:\\
					For any $k \in \{0, 1, 2, 3\}$\\
				\eIf{$W_k[3] = W_k[7]$}
					{
						${\mathcal{R}}_i[3]$ $\leftarrow 1$
					}
					{
						${\mathcal{R}}_i[3]$ $\leftarrow 0$
					}
				\eIf{${\mathcal{R}}_i[3] \neq {\mathcal{R}}_i[7]$}
					{
						${\mathcal{R}}_i[2]$ $\leftarrow 0$, ${\mathcal{R}}_i[6]$ $\leftarrow 1$\;
					}
					{
						\uIf{($W_0[2] = -1$) or ($W_1[2] = -1$) or ($W_2[2] = -2$) or ($W_3[1] = -1$) or ($W_0[6] = -1$)}
					
								{
									${\mathcal{R}}_i[2]$ $\leftarrow 0$, ${\mathcal{R}}_i[6]$ $\leftarrow 0$ \;
								}
						\uElseIf{$W_0[6] = 1$}
								{
									${\mathcal{R}}_i[2]$ $\leftarrow 1$, ${\mathcal{R}}_i[6]$ $\leftarrow 1$\;
								}
						\Else{
									${\mathcal{R}}_i[2] \leftarrow {\mathcal{R}}_i[6] \leftarrow \alpha$ where $\alpha \in \{0, 1\}$ \;
							}
						}	

\textbf{Step 8 :}\label{st8} Call {\it FindNextWeight($i$, $\Gamma_k$, $W_k$)} for each $k \in \{0, 1, 2, 3\}$\;
\textbf{Step 9 :}\label{st9} $i \leftarrow i + 1$\\
		\lIf{$i<n-2$}{goto Step 6}
		Set $\Gamma_0 \leftarrow \Gamma_0 \cap \{0, 2, 4, 6\}$, $\Gamma_1 \leftarrow \Gamma_1 \cap \{0, 2, 4, 6\}$, $\Gamma_2 \leftarrow \Gamma_2 \cap \{1, 3, 5, 7\}$, $\Gamma_3 \leftarrow \Gamma_3 \cap \{1, 3, 5, 7\}$\;
		Set ${\mathcal{R}}_{n-2}[0] \leftarrow 0$, ${\mathcal{R}}_{n-2}[7] \leftarrow 1$\;
		Check the bellow conditions and set the next state values of RMTs of ${\mathcal{R}}_{n-2}$\\
		 \eIf{$W_k[4] = W_k[0]$, $k \in \{0, 1\}$}
				{
					${\mathcal{R}}_{n-2}[4]$ $\leftarrow 0$ \;
				}
				{
					${\mathcal{R}}_{n-2}[4]$ $\leftarrow 1$ \;
				}
				
		 \eIf{${\mathcal{R}}_{n-2}[4] \neq {\mathcal{R}}_{n-2}[0]$}
				{
					${\mathcal{R}}_{n-2}[1]$  $\leftarrow 0$, ${\mathcal{R}}_{n-2}[5]$  $\leftarrow 1$\;
				}
				{
					\uIf{$W_2[1]=W_2[5]=-1$}
						{
							${\mathcal{R}}_{n-2}[1]$  $\leftarrow 0$, ${\mathcal{R}}_{n-2}[5]$  $\leftarrow 0$\;
						}
					\uElseIf{$W_3[1]=W_3[5]=1$}
						{
							${\mathcal{R}}_{n-2}[1]$  $\leftarrow 1$, ${\mathcal{R}}_{n-2}[5]$  $\leftarrow 1$\;
						}
					
					\Else{${\mathcal{R}}_{n-2}[1] \leftarrow {\mathcal{R}}_{n-2}[5] \leftarrow \alpha$ where $\alpha \in \{0, 1\}$ \;}
				}

		 \eIf{$W_k[3] = W_k[7]$, $k \in \{2, 3\}$}
				{
					${\mathcal{R}}_{n-2}[3]$ $\leftarrow 1$ \;
				}
				{
					${\mathcal{R}}_{n-2}[3]$ $\leftarrow 0$ \;
				}
				
		 \eIf{${\mathcal{R}}_{n-2}[3] \neq {\mathcal{R}}_{n-2}[7]$}
				{
					${\mathcal{R}}_{n-2}[2]$  $\leftarrow 0$, ${\mathcal{R}}_{n-2}[6]$  $\leftarrow 1$\;
				}
				{
					\uIf{$W_0[2]=W_0[6]=-1$}
						{
							${\mathcal{R}}_{n-2}[2]$  $\leftarrow 0$, ${\mathcal{R}}_{n-2}[6]$  $\leftarrow 0$\;
						}
					\uElseIf{$W_1[2]=W_1[6]=1$}
						{
							${\mathcal{R}}_{n-2}[2]$  $\leftarrow 1$, ${\mathcal{R}}_{n-2}[6]$  $\leftarrow 1$\;
						}
					
					\Else{${\mathcal{R}}_{n-2}[2] \leftarrow {\mathcal{R}}_{n-2}[6] \leftarrow \alpha$ where $\alpha \in \{0, 1\}$ \;}
				}
	Call {\it FindNextWeight($n-2$, $\Gamma_k$, $W_k$)} for each $k \in \{0, 1, 2, 3\}$\;

\textbf{Step 10 :}\label{st10} Set $\Gamma_0 \leftarrow \Gamma_0 \cap \{0, 4\}$, $\Gamma_1 \leftarrow \Gamma_1 \cap \{1, 5\}$, $\Gamma_2 \leftarrow \Gamma_2 \cap \{2, 6\}$, $\Gamma_3 \leftarrow \Gamma_3 \cap \{3, 7\}$\;

		Set ${\mathcal{R}}_{n-1}[0] \leftarrow 0$, ${\mathcal{R}}_{n-1}[7] \leftarrow 1$\;
		\eIf{$W_0[4] = 0$}
			{
				${\mathcal{R}}_{n-1}[4] \leftarrow 0$\;
			}
			{
				${\mathcal{R}}_{n-1}[4] \leftarrow 1$\;
			}
		\eIf{$W_3[3] = 0$}
			{
				${\mathcal{R}}_{n-1}[3] \leftarrow 1$\;
			}
			{
				${\mathcal{R}}_{n-1}[3] \leftarrow 0$\;
			}
		For each $r \in \{1, 5\}$\\	
		\eIf{$W_1[r] = 0$}
			{
				${\mathcal{R}}_{n-1}[r] \leftarrow 0$\;
			}
			{
				${\mathcal{R}}_{n-1}[r] \leftarrow 1$\;
			}
		For each $r \in \{2, 6\}$\\
		\eIf{$W_2[r] = 0$}
			{
				${\mathcal{R}}_{n-1}[r] \leftarrow 1$\;
			}
			{
				${\mathcal{R}}_{n-1}[r] \leftarrow 0$\;
			}

			Call {\it FindNextWeight($n-1$, $\Gamma_k$, $W_k$)} for each $k \in \{0, 1, 2, 3\}$\;

\textbf{Step 11 :}\label{st11} Output the NCCA rule vector ${\mathcal{R}}= \langle{\mathcal{R}}_0, {\mathcal{R}}_1, \cdots, {\mathcal{R}}_{n-1}\rangle$\;

\caption{\emph{Synthesize NCCA contd...}}
\label{Syn_NCCA}
\end{algorithm}
 
The steps of the algorithm is noted in Algorithm~\ref{Syn_NCCA} which, like Algorithm~\ref{analysisNCCA}, also develops the {\em super node} of a level $i$, and finds the weights of RMTs of the node. Individual rules of the CA, to be synthesized, are selected primarily based on these weights. As input, the algorithm takes $n \geq 5$, the CA size, and outputs an $n$-cell NCCA. The algorithm uses two data structures: $\Gamma_k$ -- to store the RMTs present in $\Gamma_k^i$ of level $i$, and $W_k$ -- to store the weight of RMTs in $\Gamma_k$ for each $k \in \{0, 1, 2, 3\}$.

In Step 1, the algorithm forms the root node ${\mathcal{N}}_0$ (which is also a super node of level 0) of the tree, and initializes the weights of RMTs. In Step 2, we set ${\mathcal{R}_i}[7] \leftarrow 1$ and ${\mathcal{R}_i}[0] \leftarrow 0$ for each rule ${\mathcal{R}_i}$ of rule vector ${\mathcal{R}}$. After that, we randomly select a number conserving rule as ${\mathcal{R}}_{0}$ and find the weights of the RMTs, which are present in ${\mathcal{N}}_1$ (Step 3). To find the weights of RMTs, we use the procedure, already developed in Section~\ref{algorithms}. By Step 4, however we find out ${\mathcal{R}}_{1}$, then the weights of RMTs of ${\mathcal{N}}_2$ are found out (Step 5). To find out ${\mathcal{R}}_{2}$ to ${\mathcal{R}}_{n-3}$ ($i > 1$) we repeat Steps 6, 7 and 8. To find out ${\mathcal{R}}_{n-2}$ and ${\mathcal{R}}_{n-1}$, the algorithm uses Step 9 and Step 10.

\noindent {\bf Complexity:} The time requirement of Algorithm~\ref{Syn_NCCA} depends on $n$, the size of CA only. Obviously time complexity of Algorithm~\ref{Syn_NCCA} is $O(n)$.

\begin{theorem}
\label{correct_Analysis}
Algorithm~\ref{Syn_NCCA} correctly synthesizes a NCCA of size $n \geq 5$.
\end{theorem}

\begin{proof}
Correctness of the algorithm follows the correctness of the theories, developed in previous sections. Because, all the rules, $\langle{\mathcal{R}}_0, {\mathcal{R}}_1, \cdots, {\mathcal{R}}_i, \cdots, {\mathcal{R}}_{n-1}\rangle$ of a rule vector of size $n$ are chosen utilizing those theories, by the algorithm.
\end{proof}

\begin{example}
\label{examSyn}
Let us synthesize a rule vector of size 6 using the Algorithm~\ref{Syn_NCCA}. The root node ${\mathcal{N}}_{0}$ is formed and assigned the weight of each RMTs to 0 (Step 1). In Step 2, we assign the next state value of RMT 7 to 1 and RMT 0 to 0 for each ${\mathcal{R}}_{i}$ ($i \in \{0, \cdots, n-1\}$). After that we randomly select a rule as ${\mathcal{R}}_{0}$ from the rule set \{136, 170, 184, 192, 204, 226, 238, 240, 252\}. Suppose, ${\mathcal{R}}_{0} = 192$. Then, find out the weights of RMTs of ${\mathcal{N}}_{1}$.

To find out the ${\mathcal{R}}_{1}$ we check the conditions of Step 4. From node ${\mathcal{N}}_{1}$ of Fig.~\ref{algo_steps}, we see that $W_0[2]=0$, $W_3[5]=0$, $W_2[3]=0$ and $W_1[4]=1$, so next state values of RMT 3 is 1. Now if consider ${\mathcal{R}}_{1}[2] = 0$ then using the other conditions, let us set the following: ${\mathcal{R}}_{1}[4] \leftarrow {\mathcal{R}}_{1}[6] \leftarrow 0$ and ${\mathcal{R}}_{1}[1] \leftarrow {\mathcal{R}}_{1}[5] \leftarrow 0$, which result in rule 136 as ${\mathcal{R}}_{1}$.

To select ${\mathcal{R}}_{2}$ and ${\mathcal{R}}_{3}$, Step 6, Step 7 and Step 8 are used. Using the conditions developed in these steps, one can choose ${\mathcal{R}}_{2} = 184$ and ${\mathcal{R}}_{3} = 252$. However, to find ${\mathcal{R}}_{4}$, we need to use Step 9, because level 4 is the ($n-2$)$^{th}$ level. From node ${\mathcal{N}}_{4}$, we see that $W_k[4]=W_k[0]$ and $W_k[3]=W_k[7]$ for any $k$. So next state values of RMT 4 and RMT 3 are set as 0 and 1 respectively. Using the other conditions, let us set the following: ${\mathcal{R}}_{4}[2] \leftarrow {\mathcal{R}}_{4}[6] \leftarrow 1$ and ${\mathcal{R}}_{4}[1] \leftarrow {\mathcal{R}}_{4}[5] \leftarrow 0$, which result in rule 204 as ${\mathcal{R}}_{4}$. Similarly, to find ${\mathcal{R}}_{5}$, we need to use Step 10, because level 5 is the ($n-1$)$^{th}$ level. From node ${\mathcal{N}}_{5}$, we see that $W_0[4] = 0$ and $W_3[3] = 0$. So the next state values of RMT 4 and RMT 3 are set as 0 and 1 respectively. Again from node ${\mathcal{R}}_{5}$, we see that $W_1[1] = W_1[5] \neq 0$ and $W_2[2] = W_2[6] = 0$. So ${\mathcal{R}}_{5}[2] \leftarrow {\mathcal{R}}_{5}[6] \leftarrow 1$ and ${\mathcal{R}}_{5}[1] \leftarrow {\mathcal{R}}_{5}[5] \leftarrow 1$ which result in rule 238 as ${\mathcal{R}}_{5}$.

Finally, Step 11 return the output, the NCCA rule vector ${\mathcal{R}}= \langle{192, 136, 184, 252, 204, 238}\rangle$.
\end{example}

\section{Conclusion}

This paper has studied the one dimensional two-state 3-neighborhood non-uniform number conserving cellular automata (NCCAs). The reachability tree has been utilized to do such study. Utilizing the reachability tree, a number of theorems, corollaries have been developed. Here, we have identified that only 9 rules out of the 256 Wolfram's CA rules can take part in non-uniform NCCAs. Finally, we report two algorithms, for deciding and synthesizing non-uniform NCCAs of size $n$. Both the algorithms run in O($n$) time.

\bibliographystyle{alpha}
\bibliography{References}

\end{document}